\RequirePackage{amsmath}
\documentclass[12pt]{elsarticle}
\usepackage[utf8]{inputenc}
\usepackage[T1]{fontenc}
\usepackage{latexsym}
\usepackage{fix-cm}
\usepackage{amssymb}
\usepackage{array}

\usepackage{url}
\usepackage{hyperref}
\usepackage{lineno}
\usepackage{graphicx}
\usepackage{tabularx}
\usepackage{array}
\usepackage{mdwmath}
\usepackage{mdwtab}
\usepackage{float}
\usepackage{stmaryrd}
\usepackage{eulervm}
\usepackage{palatino}
\usepackage[cmyk]{xcolor}
\usepackage{tikz}
\usetikzlibrary{shapes,arrows}
\usepackage[utf8]{inputenc}
\usepackage[T1]{fontenc}
\usepackage{latexsym}

\usepackage{amssymb}
\usepackage{makecell}
\usepackage{diagbox}

\usepackage{graphicx}
\usepackage{tabularx}
\usepackage{array}

\usepackage{booktabs}

\usepackage{algorithm2e}
\newcommand{\pc}{\mathbf{P}}
\newcommand{\oc}{\mathbf{O}}

\newcommand{\pp}{\mathbb{P}}

\newcommand{\kap}{\kappa}

\newtheorem{theorem}{Theorem}
\newdefinition{definition}{Definition}
\newtheorem{proposition}{Proposition}
\newproof{proof}{Proof}
\newtheorem{example}{Example}
\newtheorem{problem}{Problem}
\newtheorem{remark}{Remark}

\journal{Forthcoming}

\begin{document}

\begin{frontmatter}

%\setcounter{page}{1}
%\begin{frontmatter}
 
\title{High Granular Operator Spaces, and Less-Contaminated General Rough Mereologies}
\author{\textsf{A. Mani}}
\address{International Rough Set Society\\
9/1B, Jatin Bagchi Road\\
Kolkata(Calcutta)-700029, India\\
Email: \texttt{$a.mani.cms@gmail.com$}\\
Homepage: \url{http://www.logicamani.in}}
%\ead{a.mani.cms@gmail.com}}

\maketitle

\begin{abstract}

Granular operator spaces and variants had been introduced and used in theoretical investigations on the foundations of general rough sets by the present author over the last few years. In this research, higher order versions of these are presented uniformly as partial algebraic systems. They are also adapted for practical applications when the data is representable by data table-like structures according to a minimalist schema for avoiding contamination. Issues relating to valuations used in information systems or tables are also addressed. The concept of contamination introduced and studied by the present author across a number of her papers, concerns mixing up of information across semantic domains (or domains of discourse). Rough inclusion functions (\textsf{RIF}s), variants, and numeric functions often have a direct or indirect role in contaminating algorithms. Some solutions that seek to replace or avoid them have been proposed and investigated by the present author in some of her earlier papers. Because multiple kinds of solution are of interest to the contamination problem, granular generalizations of RIFs are proposed, and investigated. Interesting representation results are proved and a core algebraic strategy for generalizing Skowron-Polkowski style of rough mereology (though for a very different purpose) is formulated. A number of examples have been added to illustrate key parts of the proposal in higher order variants of granular operator spaces. Further algorithms grounded in mereological nearness, suited for decision making in human-machine interaction contexts, are proposed by the present author.  Applications of granular \textsf{RIF}s to partial/soft solutions of the inverse problem are also invented in this paper. 
\end{abstract}


\end{frontmatter}

%\linenumbers

\section{Introduction}

From a theoretical perspective, this research is about improvements in the foundations of granular operator spaces, granular generalizations of general rough inclusion functions, their properties and role in algorithms from a contamination-avoidance perspective. This means that a number of questionable assumptions \cite{am240} are not used in related derivations and definitions are \emph{strictly granular} relative to the axiomatic approach. 

From an application perspective, this paper relates to generalized numeric measures (relating to descriptive nearness, rough mereology, proximity, quality of classification and others) for data from human reasoning that can only be partly captured by data tables, and mixed data sets in which the attribute set is not known in a proper way in the first place. An application to decision-making in human-machine interaction contexts (that can also be used for robot navigation in environments exhibiting diverse transient features) is also proposed by the present author.

If $A, B\in\mathcal{S}\subseteq \wp(S)$, with $\mathcal{S}$ being closed under intersection, $S$ being a finite set and if $\# ()$ is the cardinality function, then the quantity 
\begin{equation*}\label{rif0}
\nu(A, B) = \left\lbrace  \begin{array}{ll}
 \dfrac{\# (A\cap B)}{\# (A)} & \text{if } A\neq \emptyset\\
 1 & \text{if } A= \emptyset\\
 \end{array} \right. \tag{K0}                                                                                                             
\end{equation*}
can be interpreted in multiple ways including as rough inclusion function, conditional subjective probability, relative degree of misclassification, majority inclusion function and inclusion degree. In this it is possible to replace \emph{intersections} with commonality operations that need not be idempotent, commutative or even associative. Many generalizations of this function are known in the rough set, belief theory, subjective probability, fuzzy set and ML literature. It leads to ideas of concepts being close or similar to each other in the contexts of rough sets -- but subject to a number of hidden conditions. At the same time this is at conflict with concepts of association reducts in rough sets \cite{sdrskt06,amst96}. A major problem with such functions is that they are not a function of granules (in multiple perspectives) \cite{am240,am501}.

Generalized versions of rough inclusion functions have also been used as primitives for defining concepts of general approximation spaces in \cite{ss2010}. These are particularly significant for one of the rough mereological approaches \cite{ps3,lp2011}. This is improved in a sense and generalized in another direction in this research.  An overview of generalizations of rough inclusion functions can be found in \cite{ag3,ag2009}. Theoretical studies on such functions in the context of rough sets is limited. Some connections with other numeric functions of rough sets can be found in \cite{jzd2003}.

Data analysis maybe intrusive (invasive) or non-intrusive relative to the assumptions made on the dataset used in question \cite{gdu}. Non-invasive data analysis was defined in a vague way in \cite{gdu} as one that 
\begin{itemize}
\item {Is based on the idea of \emph{data speaking for themselves},}
\item {Uses minimal model assumptions by drawing all parameters from the observed data, and}
\item {Admits ignorance when no conclusion can be drawn from the data at hand.}
\end{itemize}
Key procedures deemed to be non-invasive are data discretization (or horizontal compression),\\ randomization procedures, reducts of various kinds within rough set data analysis, and rule discovery with the aid of maximum entropy principles. 

In general approaches to rough sets, some types of objects like those that are roughly equivalent, or those predicated with approximations or more complex objects may be studied \cite{am240,amdsc2016}. These in turn with associated meta operations and rules correspond to semantic domains (or domains of discourse). As mentioned in \cite{amdsc2016} by the present author: In classical rough sets \cite{zpb}, an approximation space is a pair of the form $\left\langle S,\,R \right\rangle $, with $R$ being an equivalence on the set $S$. On the power set $\wp (S)$, lower and upper
approximation operators, apart from the usual Boolean operations, are definable. The resulting structure constitutes a semantics for rough sets (though not satisfactory) from a classical perspective. This continues to be true even when $R$ is some other type of binary relation. More generally it is possible to replace $\wp (S)$ by some set with a parthood relation and some approximation operators defined on it. The associated semantic domain in the sense of a collection of restrictions on possible objects, predicates, constants, functions and low level operations on those will be referred to as the classical semantic domain (meta-C) for general rough sets. In contrast, the semantics associated with sets of roughly equivalent or relatively indiscernible objects with
respect to this domain is a rough semantic domain (meta-R). Many other domains, including hybrid semantic domains, can be generated. These have been used in \cite{am99,am105,am3} for different types of rough semantics. In \cite{am105}, the reasoning is within the power set of the collection of possible order-compatible partitions of the set of roughly equivalent elements. The concept of \emph{semantic domain} is therefore similar to the sense in which it is used in general abstract model theory \cite{md} (though one can object to formalization on different philosophical grounds).

The concept of \emph{contamination} was introduced in \cite{am99} and explored in \cite{am240,am501,am9114} by the present author. It is associated with a distinct minimalist approach that takes the semantic domains involved into account and has the potential to encompasses the three principles of non-intrusive analysis. Some sources of contamination are those arising from assumptions about distribution of attributes, introduction of assumptions valid in one semantic domain into another by oversight \cite{am3600}, numeric functions used in rough sets (and soft computing in general) and fuzzy representation of linguistic hedges. \emph{The contamination problem in simplified terms is that of reducing artificial constructs in general rough sets towards capturing essential rough reasoning at that level}. Reduction of contamination is relevant in all model/algorithm building contexts of formal approaches to vagueness. In particular, it includes the problem of minimizing or eliminating the contamination of models of meta-R or other rough semantic domain fragments by meta-C aspects.  The concept is essential for modeling relation between attributes \cite{am9501,am3930}. A Bayesian approach to modeling causality between attributes is proposed in \cite{ndwhz2019} -- the approach tries to avoid contamination to an extent.

Granules or information granules are often the minimal discernible concepts that can be used to construct all relatively crisp complex concepts in a vague reasoning context. Such constructions typically depend on substantial assumptions made by the framework employed in question \cite{am240,am501,am9411,tyl}. Major granular computing approaches can be classified into 
\begin{itemize}
\item {Primitive Granular Computing Paradigm: PGCP (see \cite{am501})}
\item {Classical Granular Computing Paradigm: CGCP including\\ GrC model approach \cite{tyl}, and definite object approach \cite{yzm2012,hmy2019}.}
\item {Axiomatic Granular Computing Paradigm: AGCP due to the present author}
\end{itemize}

In the present author's axiomatic approach to granularity\\ \cite{am240,am9114,am9411,am501,am3930,am3600}, fundamental ideas of non-intrusive data analysis have been critically examined and methods for reducing contamination of data (through external assumptions) have been proposed. The need to avoid over-simplistic constructs like rough membership and inclusion functions have been stressed in the approach by her. New granular measures that are compatible with rough domains of reasoning, and granular correspondences that avoid measures have also been invented in the papers. These granular measures can be improved further with the following goals in mind:
\begin{itemize}
\item {\emph{to improve the basic frameworks of general rough sets so that non-Boolean combinations of valuations of attributes can be accommodated}.}
\item {to improve the measures so that they can integrate seamlessly with rough and hybrid semantic domains,}
\item {to provide a less-contaminated or contamination free measure that goes beyond the limited heuristics of dominance based rough sets,}
\item {to provide a reasonable basis for translating reasoning across different approaches to vagueness, mereology and uncertainty, and}
\item {to improve solution strategies of the inverse problem.}
\end{itemize}

In this research, higher order variants of granular operator spaces are reformulated as partial algebraic systems, and variants of granular operator spaces suited for applications are introduced -- these have additional facilities for handling labeling, imposing meta level constraints and \emph{handling non-Boolean combinations of attribute valuations}.  The inverse problem, considered in earlier papers by the present author, is reformulated in the improved formulation as a representation (or possibly a duality) problem. Granular RIFs are also used in possible solution strategies for the same. Various examples have been constructed to demonstrate aspects of the formalism. 

A new class of less contaminated matrix measures of rough\\ inclusion (termed granular rough inclusion functions GRIFs) is also introduced, studied and applied to a new model for decision making in human-machine interface contexts. The measure and algorithm can be read as a radical generalization of parts of Skowron-Polkowski style mereology based on a distinct minimalist perspective, and also as \emph{apparent parthood} \cite{am9969}. The exact scope of the proposed measure is also explored from a mathematical perspective and further extensions are motivated. The proposal provides for a direct comparison between GRIFs and the original parthood within the same system under additional meta level constraints. Other possible applications to local rough sets, big data, sociological datasets and linguistics are also discussed in brief. The nature of connections with granular and general RIFs in the context of relation based, neighborhood based and cover based rough sets are also demonstrated.   

In the next section, relevant background and notation is presented in brief. Higher order variants of granular operator spaces are reformulated as partial algebraic systems in the third section. Many examples are provided in the section. Concrete versions of higher order variants of granular operator spaces are presented as tabular variants that generalize information systems from a higher order perspective in the following section, and key assumptions about however information tables (or systems) are questioned. In the fifth section, aspects of general rough inclusion functions are adapted for \textsf{set HGOS} and new operations are proposed over them. Granular rough inclusion functions are defined, operations over them 
studied, results proved on their representation, and new mereologies outlined in the sixth section. These functions are extended to more general abstract high granular operator spaces and variants in the following section. In the same section, these are applied to inverse problems. Applications to human-machine interfaces including the Pilot's algorithm is proposed, and applications to other linguistic contexts are part of the next section. Finally, open problems are proposed, and concluding remarks made in the following section.     

\section{Background and Notation}\label{bck}

For basics of rough sets, the reader is referred to \cite{ppm2,lp2011}. The granular approach due to the present author can be found in \cite{am501,am240,am9114}. The reader is also expected to have a reasonable understanding of reducts and related computation \cite{mopizi,js09}. In the set-theoretic formalisms of the present paper, a \emph{granulation} is a subset of a powerset that satisfies some conditions. In \cite{tyl}, a model of granulation (GrC model) is defined as a collection of objects of a category along with a set of n-ary relations (subobject of a product object). The present usage is obviously distinct (though related).

\emph{Throughout the paper, quantification is enclosed in braces for easier reading;\\
$(\forall a_1,\ldots , a_n)$ is the same as $\forall a_1,\ldots ,a_n$}.

\subsection{Information Tables}
The concept of \emph{information} can also be defined in many different and non-equivalent ways. In the present author's view \emph{anything that alters or has the potential to alter a given context in a significant positive way is information}. In the contexts of general rough sets, the concept of information must have the following properties:
\begin{itemize}
\item {information must have the potential to alter supervenience relations in the contexts (A set of properties $Q$ supervene on another set of properties $T$ if there exists no two objects that differ on $Q$ without differing on $T$),}
\item {information must be formalizable and }
\item {information must generate concepts of roughly similar collections of properties or objects.}
\end{itemize}

Because rough inclusion functions and variants may be interpreted as conditional subjective probabilities, the nature of information in related contexts are relevant. In the present author's opinion,  \emph{information} in the contexts of subjective probability should:
\begin{itemize}
\item {have the potential to alter uncertainty relations in the context,}
\item {be formalizable, }
\item {be bounded, }
\item {have associated methods for handling temporality,}
\item {be granular, and}
\item {be relativizable.}
\end{itemize}

The above can be read as a minimal set of desirable properties. In practice, additional assumptions are common in all approaches and the above is about a minimalism. This has been indicated to suggest that comparisons may work well when ontologies are justified.   

The concept of an information system or table is not essential for obtaining a granular operator space or higher order variants thereof. But it often happens that they arise from such tables. The inverse problem is essentially the problem of finding conditions that ensures this.

Information systems or more correctly, information storage and\\ retrieval systems (also referred to as information tables, descriptive systems, knowledge representation system) are basically representations of structured data tables. In the paper \cite{cd2017}, a critical reflection on the terminology used in rough sets and allied fields can be found with a suggestion to avoid plural meanings for the same term.  When columns for decision are also included, then they are referred to as \emph{decision tables}. Often rough sets arise from \emph{information tables} and decision tables. In the literature on artificial intelligence, database theory and machine learning, the term \emph{information system} refers to an integrated heterogeneous system that has components for collecting, storing and processing data. From a mathematical point of view, this can be described using heterogeneous partial algebraic systems. In rough set contexts, this generality has not been exploited as of this writing. 

An \emph{information table} $\mathcal{I}$, is a relational system of the form \[\mathcal{I}\,=\, \left\langle \mathfrak{O},\, \mathbb{A},\, \{V_{a} :\, a\in \mathbb{A}\},\, \{f_{a} :\, a\in \mathbb{A}\}  \right\rangle \]
with $\mathfrak{O}$, $\mathbb{A}$ and $V_{a}$ being respectively sets of \emph{Objects}, \emph{Attributes} and \emph{Values} respectively.
$f_a \,:\, \mathfrak{O} \longmapsto \wp (V_{a})$ being the valuation map associated with attribute $a\in \mathbb{A}$. Values may also be denoted by the binary function $\nu : \mathbb{A} \times \mathfrak{O} \longmapsto \wp{(V)} $ defined by for any $a\in \mathbb{A}$ and $x\in \mathfrak{O}$, $\nu(a, x) = f_a (x)$.

An information table is \emph{deterministic} (or complete) if
\[(\forall a\in At)(\forall x\in \mathfrak{O}) f_a (x) \text{ is a singleton}.\] It is said to be \emph{indeterministic} (or incomplete) if it is not deterministic that is
\[(\exists a\in At)(\exists x\in \mathfrak{O}) f_a (x) \text{ is not a singleton}.\]

Relations may be derived from information tables by way of conditions of the following form: For $x,\, w\,\in\, \mathfrak{O} $ and $B\,\subseteq\, \mathbb{A} $, $(x,\,w)\,\in\, \sigma $ if and only if $(\mathbf{Q} a, b\in B)\, \Phi(\nu(a,\,x),\, \nu (b,\, w),) $ for some quantifier $\mathbf{Q}$ and formula $\Phi$. The relational system $S = \left\langle \underline{S}, \sigma \right\rangle$ (with $\underline{S} = \mathbb{A}$) is said to be a \emph{general approximation space}. 

In particular if $\sigma$ is defined by the condition Equation..\ref{pawl}, then $\sigma$ is an equivalence relation and $S$ is referred to as an \emph{approximation space}.
\begin{equation*}\label{pawl}
(x, w)\in \sigma \text{ if and only if } (\forall a\in B)\, \nu(a,\,x)\,=\, \nu (a,\, w) 
\end{equation*}

In classical rough sets, on the power set $\wp (S)$, lower and upper
approximations of a subset $A\in \wp (S)$ operators, apart from the usual Boolean operations, are defined as per: 
\[A^l = \bigcup_{[x]\subseteq A} [x] \; ; \; A^{u} = \bigcup_{[x]\cap A\neq \varnothing } [x],\,\]
with $[x]$ being the equivalence class generated by $x\in S$. If $A, B\in \wp (S)$, then $A$ is said to be \emph{roughly included} in $B$ $(A\sqsubseteq B)$ if and only if $A^l \subseteq B^l$ and $A^u\subseteq B^u$. $A$ is roughly equal to $B$ ($A\approx B$) if and only if $A\sqsubseteq B$ and $B\sqsubseteq A$. The positive, negative and boundary region determined by a subset $A$ are respectively $A^l$, $A^{uc}$ and $A^{u}\setminus A^l$ respectively.

Given a fixed $A\in \wp(S)$, a \emph{Rough membership function} is a map $f_A: S \longmapsto [0,1]$ that are defined via \[(\forall x)\, f_A(x) = \dfrac{card([x]\cap A)}{card([x])}.\] Rough membership functions can be generalized to other general rough structures but lose many of the better properties valid in the classical context.

A \emph{cover} $\mathcal{C}$ on a set $\underline{S}$ is any sub-collection of $\wp(\underline{S})$. It is said to be \emph{proper} just in case $\bigcup{\mathcal{C}} = \underline{S}$. The tuple $\mathfrak{C} =\left\langle\underline{S},\, \mathcal{C}   \right\rangle$ is said to be a \emph{covering approximation space}. In this paper all covering approximation spaces will be assumed to be proper unless stated otherwise.

A \emph{neighborhood operator} $n$ on a set $\underline{S}$ is any map of the form $n:\, \underline{S}\longmapsto \wp{\underline{S}}$. It is said to be 
\emph{Reflexive} if \begin{equation}(\forall x\in \underline{S}) \, x\in n(x) \tag{Nbd:Refl}\end{equation}
The collection of all neighborhoods $\mathcal{N}= \{n(x) \, : x\in \underline{S}\}$ of $\underline{S}$ will form a cover if and only if $(\forall x)(\exists y) x\in n(y)$ (anti-seriality). So in particular a reflexive relation on $\underline{S}$ is sufficient to generate a proper cover on it. Of course, the converse association does not necessarily happen in a unique way.  

If $\mathcal{S}$ is a cover of the set $\underline{S}$, then the \emph{neighborhood} of $x\in \underline{S}$ is defined via, \begin{equation}nbd(x)\,=\,\bigcap\{K:\,x\in K\,\in\,\mathcal{S}\} \tag{Cover:Nbd}\end{equation} 

The \emph{maximal description} of an element $x\in \underline{S}$ is defined to be the collection:
\begin{equation}\mathrm{MD} (x)\,=\,\{A\,:x\in A\in \mathcal{S},\, (\forall{B \in \mathcal{S}})(x\in B\rightarrow
\sim(A\subset B))\} \tag{Cover:MD}\end{equation}

The \emph{indiscernibility} (or friends) of an element $x\,\in \underline{S}$
is defined to be \begin{equation}Fr(x)\,=\,\bigcup\{K:\,x\in K\in\mathcal{S}\} \tag{Cover:FR}\end{equation}
An element $K\in \mathcal{S}$ is said to be \emph{reducible} if and only if \begin{equation}(\forall x\in K)
K\neq\,MD(x) \tag{Cover:Red}\end{equation} The collection $\{nbd(x):\,x\in\,S\}$ will be denoted by $\mathcal{N}$. The cover obtained by the removal of all reducible elements is called a \emph{covering reduct}.

Boolean algebra with approximation operators constitutes a semantics for classical rough sets (though not satisfactory). This continues to be true even when $R$ in the approximation space is replaced by any binary relation. More generally it is possible to replace $\wp (S)$ by some set with a part-hood relation and some approximation operators defined on it \cite{am240}. The associated semantic domain in the sense of a collection of restrictions on possible
objects, predicates, constants, functions and low level operations on those is 
referred to as the classical semantic domain for general rough sets. In contrast, the semantic domain
associated with sets of roughly equivalent or relatively indiscernible objects with
respect to this domain is a \emph{rough semantic domain}. Actually many other
semantic domains, including hybrid semantic domains, can be generated and have been used for example in choice-inclusive semantics \cite{am99}, but these two broad domains will always be - though not necessarily with a nice correspondence between the two.

\subsection{Algebraic Concepts}

A \emph{semiring} is an algebra of the form $A = \left\langle \underline{A},\, +, \, \cdot,\, 0  \right\rangle$ with $\underline{A}$ being a set,   $ + $ being a commutative monoidal operation with unit element $0$, and $ \cdot$ being an associative operation that satisfies the following distributivity conditions:
\begin{align*}
(\forall a, b, c) a\cdot (b+c) = (a\cdot b) + (a\cdot c)    \tag{l-distributivity}\\
(\forall a, b, c) (b+c)\cdot a = (b\cdot a) + (c\cdot a)    \tag{r-distributivity}
\end{align*}

If $A= (a_{ij})_{n \times m}$ and $B= (b_{ij})_{n \times m}$ are two $n\times m$ matrices over a semiring $\mathbb{F}$, then the $H$-product of the two is defined by \[A\circledast B = (a_{ij}\cdot b_{ij})_{n\times m} \]  $\circledast$ is a commutative monoidal operation that distributes over matrix addition. 

For basics of partial algebras, the reader is referred to \cite{bu,lj}.
\begin{definition}
A \emph{partial algebra} $P$ is a tuple of the form \[\left\langle\underline{P},\,f_{1},\,f_{2},\,\ldots ,\, f_{n}, (r_{1},\,\ldots ,\,r_{n} )\right\rangle\] with $\underline{P}$ being a set, $f_{i}$'s being partial function symbols of arity $r_{i}$. The interpretation of $f_{i}$ on the set $\underline{P}$ should be denoted by $f_{i}^{\underline{P}}$, but the superscript will be dropped in this paper as the application contexts are simple enough. If predicate symbols enter into the signature, then $P$ is termed a \emph{partial algebraic system}.   
\end{definition}

In this paragraph the terms are not interpreted. For two terms $s,\,t$, $s\,\stackrel{\omega}{=}\,t$ shall mean, if both sides are defined then the two terms are equal (the quantification is implicit). $\stackrel{\omega}{=}$ is the same as the existence equality (also written as $\stackrel{e}{=}$) in the present paper. $s\,\stackrel{\omega ^*}{=}\,t$ shall mean if either side is defined, then the other is and the two sides are equal (the quantification is implicit). Note that the latter equality can be defined in terms of the former as 
\[(s\,\stackrel{\omega}{=}\,s \, \longrightarrow \, s\,\stackrel{\omega}{=} t)\&\,(t\,\stackrel{\omega}{=}\,t \, \longrightarrow \, s\,\stackrel{\omega}{=} t) \]

Various kinds of morphisms can be defined between two partial algebras or partial algebraic systems of the same or even different types. If $X\, =\, \left\langle\underline{X},\,f_{1},\,f_{2},\,\ldots ,\, f_{n} \right\rangle$ and $W\, =\, \left\langle\underline{W},\,g_{1},\,g_{2},\,\ldots ,\, g_{n} \right\rangle $ are two partial algebras of the same type, then a map $\varphi \, :\, X\, \longmapsto\, W$ is said to be a 
\begin{itemize}
\item {\emph{morphism} if for each $i$, $(\forall (x_1,\, \ldots \, x_k)\,\in \, dom (f_i))$ $\varphi (f_{i}(x_1 , \ldots , \, x_k))\,=\,  g_i (\varphi(x_1),\, \ldots , \, \varphi (x_k)) $}
\item {\emph{closed morphism}, if it is a morphism and the existence of\\ $g_{i} (\varphi(x_1),\, \ldots , \, \varphi (x_k))$ implies the existence of $f_{i}(x_1 , \ldots , \, x_k)$.}
\end{itemize}

Usually it is more convenient to work with closed morphisms.

\subsection{T-Norms, S-Norms}

Triangular norms (t-norms) and s-norms (or triangular conorms) are well-known in the literature on fuzzy logic and multi criteria decision making \cite{cjb2006}. They are respectively used for expressing conjunctions and disjunctions in suitable logics. 
  
\begin{definition}
A \emph{t-norm} is an operation $\otimes:\,[0,1]^{2}\,\mapsto [0,1]$ that satisfies all of the following four conditions:
\begin{align*}
(\forall a)\, a\otimes 1\,=\,a   \tag{T0}\\
(\forall a, b)\,a \otimes b\,=\,b \otimes a  \tag{Commutativity}\\
(\forall a, b, c)(b\leq c\,\longrightarrow\,a \otimes b\leq  a \otimes c)  \tag{Monotonicity}\\
a \otimes (b \otimes c))\,=\,(a \otimes b) \otimes c  \tag{Associativity}
\end{align*}
\end{definition}
\begin{definition}
If $n :[0,1] \longmapsto [0, 1]$ is a function, consider the conditions: 
\begin{align*}
n(0) = 1 \, \&\, n(1) = 0 \tag{I}\\
(\forall a, b) (a\leq b \longrightarrow n(b) \leq n(a)) \tag{Anti-mo}\\
(\forall x) x\leq n(n(x)) \tag{Weak}\\
(\forall x) x = n(n(x)) \tag{Strong}
\end{align*}
$n$ is negation if it satisfies the first three conditions. A negation $n$ is \emph{weak} if it satisfies the condition \textsf{Weak}. It is \emph{strong or involutive} if it satisfies \textsf{Strong} in addition.   
\end{definition}
\begin{definition}
An \emph{s-norm} $ \oplus:\,[0,1]^{2}\,\mapsto [0,1]$ is a function that satisfies associativity, monotonicity, commutativity and 
\begin{equation*}
(\forall x) x \oplus 0 =x  \tag{S0}
\end{equation*}
\end{definition}

In other words, t-norms are commutative monoidal order compatible operations with unit $1$ on the unit interval $[0, 1]$. s-norms are also commutative monoidal order compatible operations on $[0, 1]$, but with unit $0$. t-norms and s-norms of a finite sequence of numbers $(a_k)$ will be denoted by $\bigotimes a_k$ and $\bigoplus a_k$ respectively.

The following are popular pairs of t- and s-norms:
\begin{itemize}
\item {$ a \otimes_M b) = min(a, b)$ (Min t-norm); $ a \oplus_M b = Max(a, b)$ (Max s-norm)}
\item {$ a \otimes_P b) = a\cdot b$ (Product t-norm) ; $ a \oplus_P b = a+b -a\cdot b$ (Probabilist s-norm)}
\item {$ a \otimes_L b) = \max (a+b -1, 0) $ (\L{}ukasiwicz t-norm); $ a \oplus_L b = min(a+b, 1)$ (\L{}ukasiwicz s-norm)}
\end{itemize}
For every s-norm $\oplus$, there exists a t-norm $\otimes$ and an strong negation $n$ such that: \[(\forall a, b) a \oplus b\,=\,n( n(a) \otimes n(b)).\]  In this situation, $(\otimes, s)$ is said to be a \emph{t-norm -- s-norm pair}.

A t-norm $\otimes$ is said to be \emph{left continuous at a point} $(a, b)$ if and only if
\begin{gather*}(\forall \epsilon >0)(\exists \delta >0)(\forall (x,z))((a-\delta, b - \delta) \leq (x, z) \leq (a, b) \longrightarrow \\  |(x \otimes z) - (a \otimes b)| < \epsilon)
\end{gather*}

Any left continuous t-norm $\otimes$ can be used to define unique residual implications:
$a \implies_\otimes b = Sup \{c: \, c\otimes a \leq b\}$

\subsection{Meta Explanation of Terms}

This list is to help with reading about general rough sets, granular operator spaces and RYS (and especially about the present author's usage).

\begin{itemize}
\item {\textsf{Crisp Object}:  That which has been designated as \emph{crisp} or is an approximation of some other object.}
\item {\textsf{Vague Object}: That whose approximations do not coincide with the object or that which has been designated as a \emph{vague} object.}
\item {\textsf{Discernible Object}: That which is available for computations in a rough semantic domain (in a contamination avoidance perspective). }
\item {\textsf{Rough Object}: Many definitions and representations are possible relative to the context. From the representation point of view these are usually functions of definite or crisp objects.}
\item {\textsf{Definite Object}: An object that is invariant relative to an approximation process. In actual semantics a number of concepts of definiteness is possible. In some approaches, as in \cite{yzm2012,hmy2019}, these are taken as granules. Related theory has a direct connection with closure algebras and operators as indicated in \cite{am501}.}
\end{itemize}

\section{Variants of Granular Operator Spaces}

Granular operator spaces and related variants can be constructed from records of human reasoning, databases or from partial semantics of general rough sets. They are mathematically accessible powerful abstractions for handling semantic questions, formulation of semantics and the inverse problem. As many as six variants of such spaces have been defined by the present author -- these can be viewed as special cases of a set theoretic and a relation theoretic abstraction with abstract operations from a category theory perspective. 
Some mix-ups in terminology between higher order granular operator spaces and lower order versions have happened across earlier papers due to the present author in \cite{am9411}. Strictly speaking, higher order versions are partial algebraic systems, the \emph{space} in the terminology is because of mathematical usage conventions. In fact, in a forthcoming paper, it is shown by the present author that they are equivalent to certain single sorted partial algebras with nice interpretation.  

Rough Y-systems and granular operator spaces, introduced and studied extensively by the first author \cite{am501,am9969,am9411,am240}, are essentially higher order abstract approaches in general rough sets in which the primitives are ideas of approximations, parthood, and granularity. Motivations for the present approach relate to issues of simplifying rough Y-systems (RYS) \cite{am240,amdsc2016} to purely set theoretic contexts which in turn were motivated by the need to accommodate granulations and simultaneously generalize abstract approaches to rough sets \cite{it2,cd3,gc2018} without superfluous assumptions. But over time, the level of abstraction has evolved to cover more ground beyond general rough and fuzzy set theories. In the literature on mereology \cite{av,vie,ur,rgac15,am3930,seibtj2015}, it is argued that most ideas of binary \emph{part of} relations in human reasoning are at least antisymmetric and reflexive.  \emph{A major reason for not requiring transitivity of the parthood relation is because of the functional reasons that lead to its failure} (see \cite{seibtj2015}), and to accommodate \emph{apparent parthood} \cite{am9969}. In the context of approximate reasoning interjected with subjective or pseudo-quantitative degrees, transitivity is again not common. The role of such parthoods in higher order approaches are distinctly different from theirs in lower order approaches -- specifically, general approximation spaces of the form $S$ mentioned above with $R$ being a parthood relation are also of interest.

In a \emph{high general granular operator space} (\textsf{GGS}), introduced below, aggregation and co-aggregation operations ($\vee, \,\wedge$) are conceptually separated from the binary parthood $\pc$), and a basic partial order relation ($\leq$). Parthood is assumed to be reflexive and antisymmetric. It may satisfy additional generalized transitivity conditions in many contexts. Real-life information processing often involves many non-evaluated instances of aggregations (disjunctions), co-aggregation (conjunctions) and implications because of laziness or supporting meta data or for other reasons  -- this justifies the use of partial operations. Specific versions of a \textsf{GGS} and granular operator spaces have been studied in \cite{am501} by the present author for handling a very large spectrum of rough set contexts. \textsf{GGS} has the ability to handle adaptive situations as in \cite{skaj2016,skajsd2016} through special morphisms -- this is again harder to express without partial operations.  

The underlying set $\underline{\mathbb{S}}$ can be a set of set of attributes, but this interpretation is not compulsory. In actual practice, \textsf{the set of all attributes in a context need not be known exactly to the reasoning agent constructing the approximations. The element $\top$ may be omitted in these situations or the issue can be managed through restrictions on the granulation}. 

In real life situations, \emph{it often happens that certain objects cannot be approximated in an acceptable way}. Therefore, it can be argued that the approximations operations used should be partial. The state of affairs need not change when additional approximation operators are used. Further, the ontological commitment to totality can be huge -- for these reasons the concept of a \textsf{Pre-GGS} is also introduced below. 

\begin{definition}\label{gfsg}

A \emph{High General Granular Operator Space} (\textsf{GGS}) $\mathbb{S}$ shall be a partial algebraic system  of the form $\mathbb{S} \, =\, \left\langle \underline{\mathbb{S}}, \gamma, l , u, \pc, \leq , \vee,  \wedge, \bot, \top \right\rangle$ with $\underline{\mathbb{S}}$ being a set, $\gamma$ being a unary predicate that determines $\mathcal{G}$ (by the condition $\gamma x$ if and only if $x\in \mathcal{G}$) 
an \emph{admissible granulation}(defined below) for $\mathbb{S}$ and $l, u$ being operators $:\underline{\mathbb{S}}\longmapsto \underline{\mathbb{S}}$ satisfying the following ($\underline{\mathbb{S}}$ is replaced with $\mathbb{S}$ if clear from the context. $\vee$ and $\wedge$ are idempotent partial operations and $\pc$ is a binary predicate. Further $\gamma x$ will be replaced by $x \in \mathcal{G}$ for convenience.):

\begin{align*}
(\forall x) \pc xx \tag{PT1}\\
(\forall x, b) (\pc xb \, \&\, \pc bx \longrightarrow x = b) \tag{PT2}\\
(\forall a, b) a\vee b \stackrel{\omega}{=} b\vee a \;;\; (\forall a, b) a\wedge b \stackrel{\omega}{=} b\wedge a \tag{G1}\\
(\forall a, b) (a\vee b) \wedge a \stackrel{\omega}{=} a \; ;\; (\forall a, b) (a\wedge b) \vee a \stackrel{\omega}{=} a \tag{G2}\\
(\forall a, b, c) (a\wedge b) \vee c \stackrel{\omega}{=} (a\vee c) \wedge (b\vee c) \tag{G3}\\
(\forall a, b, c) (a\vee b) \wedge c \stackrel{\omega}{=} (a\wedge c) \vee  (b\wedge c) \tag{G4}\\
(\forall a, b) (a\leq b \leftrightarrow a\vee b = b \,\leftrightarrow\, a\wedge b = a  ) \tag{G5}\\
(\forall a \in \mathbb{S})\,  \pc a^l  a\,\&\,a^{ll}\, =\,a^l \,\&\, \pc a^{u}  a^{uu}  \tag{UL1}\\
(\forall a, b \in \mathbb{S}) (\pc a b \longrightarrow \pc a^l b^l \,\&\,\pc a^u  b^u) \tag{UL2}\\
\bot^l\, =\, \bot \,\&\, \bot^u\, =\, \bot \,\&\, \pc \top^{l} \top \,\&\,  \pc \top^{u} \top  \tag{UL3}\\
(\forall a \in \mathbb{S})\, \pc \bot a \,\&\, \pc a \top    \tag{TB}
\end{align*}

Let $\pp$ stand for proper parthood, defined via $\pp ab$ if and only if $\pc ab \,\&\,\neg \pc ba$). A granulation is said to be admissible if there exists a term operation $t$ formed from the weak lattice operations such that the following three conditions hold:
\begin{align*}
(\forall x \exists
x_{1},\ldots x_{r}\in \mathcal{G})\, t(x_{1},\,x_{2}, \ldots \,x_{r})=x^{l} \\
\tag{Weak RA, WRA} \mathrm{and}\: (\forall x)\,(\exists
x_{1},\,\ldots\,x_{r}\in \mathcal{G})\,t(x_{1},\,x_{2}, \ldots \,x_{r}) =
x^{u},\\
\tag{Lower Stability, LS}{(\forall a \in
\mathcal{G})(\forall {x\in \underline{\mathbb{S}}) })\, ( \pc ax\,\longrightarrow\, \pc ax^{l}),}\\
\tag{Full Underlap, FU}{(\forall
x,\,a \in\mathcal{G})(\exists
z\in \underline{\mathbb{S}}) )\, \pp xz,\,\&\,\pp az\,\&\,z^{l}\, =\,z^{u}\, =\,z,}
\end{align*}
\end{definition}
\emph{The conditions defining admissible granulations mean that every approximation is somehow representable by granules in a algebraic way, that every granule coincides with its lower approximation (granules are lower definite), and that all pairs of distinct granules are part of definite objects (those that coincide with their own lower and upper approximations).} Special cases of the above are defined next.

\begin{definition}
\begin{itemize}
\item {In a \textsf{GGS}, if the parthood is defined by $\pc ab$ if and only if $a \leq b$ then the \textsf{GGS} is said to be a \emph{high granular operator space} \textsf{GS}.}
\item {A \emph{higher granular operator space} (\textsf{HGOS}) $\mathbb{S}$ is a \textsf{GS} in which the lattice operations are total.}
\item {In a higher granular operator space, if the lattice operations are set theoretic union and intersection, then the \textsf{HGOS} will be said to be a \emph{set HGOS}. }
\end{itemize}
\end{definition}

\begin{definition}
In the context of Def. \ref{gfsg} if additional lower and upper approximation operations are present in the signature, then  the system will be referred to as an \emph{enhanced} \textsf{GGS} (\textsf{EGGS}. 
\end{definition}

\begin{definition}
In the context of Def. \ref{gfsg} if $l$ and $u$ are partial operations that satisfy \textsf{PL0, PUL0, PL1, PUL2,}, and \textsf{UL3}  instead of \textsf{UL1, UL2} and \textsf{UL3} respectively, then  the system will be referred to as a \emph{Pre-GGS }.  (universal quantifiers are omitted)
\begin{align*}
a^l \stackrel{\omega}{=}a^l \longrightarrow   \pc a^l  a   \tag{PL0}\\
a^u \stackrel{\omega}{=} a^u \longrightarrow   \pc a^{u}  a^{uu}    \tag{PU0}\\
a^{ll}\,\stackrel{\omega}{=}\,a^l    \tag{PL1}\\
\bot^l\, =\, \bot \,\&\, \bot^u\, =\, \bot \,\&\, \pc \top^{l} \top \,\&\,  \pc \top^{u} \top  \tag{UL3}\\
(\forall a \in \mathbb{S})\, \pc \bot a \,\&\, \pc a \top    \tag{TB} 
\end{align*}
\begin{gather*}
 \pc a b\,\&\, a^l \stackrel{\omega}{=}a^l \,\&\, b^l \stackrel{\omega}{=}b^l \,\&\, a^u \stackrel{\omega}{=}a^u \,\&\, b^u \stackrel{\omega}{=}b^u\\  \longrightarrow \pc a^l b^l \,\&\,\pc a^u  b^u \tag{PUL2}
\end{gather*}

\end{definition}

Analogously concepts of \emph{Pre-GS}, \emph{Pre-HGOS} and \emph{Set Pre-HGOS} are definable.

\begin{definition}
By the \emph{lu-one point partial completion} of a $\pi$-GGS $\mathbb{S}$ will be meant the partial algebraic system $\mathbb{S}^*$ on the set $\underline{\mathbb{S}}\cup \{o\}$ (with $o\notin \underline{\mathbb{S}}$) obtained after setting 
\[ x^u = \left\lbrace  \begin{array}{ll}
 x^{u} & \text{if } x^u \text{ is defined}\\
 o & \text{if } x^u \text{ is not defined}\\
 \end{array} \right.\]
\[ x^l = \left\lbrace  \begin{array}{ll}
 x^{l} & \text{if } x^l \text{ is defined}\\
 o & \text{if } x^l \text{ is not defined}\\
 \end{array} \right.\]
\end{definition}

In general, $\mathbb{S}^*$ need not be a \textsf{GGS}. Even over a lu-one point partial completion of a Pre-GS, it is  possible to define an equivalence relative to which, the quotient is a \textsf{GS}.

\begin{proposition}
On the lu-one point partial completion $\mathbb{S}^*$ of a finite Pre-GS $\mathbb{S}$, it is possible to define an equivalence $\sigma$ such that the quotient $\mathbb{S}^*|\sigma$ with induced operations and predicates is a \textsf{GS}.  
\end{proposition}

\begin{proof}
The partial operations $l$ and $u$ are required to be monotonic with respect to $\leq$ (that coincides with the parthood predicate $\pc$ ). 

For the completion of $l$ and $u$ to satisfy the conditions \textsf{UL1, UL2} and \textsf{UL3}, it is necessary that 

If $\pc ab \, \&\, b^l = o \, \&\, b^u = o $ and $\neg a^l = o$, then identifying $b$ with $o$ (that is require that $\sigma b o$) is necessary for induced operations and predicates to be compatible with monotonicity. If $e$ is another element such that $\pc be \, \& \, \neg e^l = o \, \&\, \neg  e^u = o$, then $\pc b^l e^l$. For monotonicity of l, u to hold, it would be necessary to identify $e, \, b, \, o$.   
 
This means $\sigma$ and operations on the quotient should be defined as per the following:

\begin{itemize}
\item {Find the set of minimal elements $H$ for which the lower or upper approximation is $o$ (finiteness ensures the existence of $H$).  }
\item {Define $\sigma xz$ if and only $\pc xz \, \&\, x\in H$ or $x, z\in H $.}
\item {On $\mathbb{S}^*|\sigma$, define $[o]^l = [o]^u = [o]$}
\item {On $\mathbb{S}^*|\sigma$, define $\pc x [o]$ for all $x$}
\item {On $\mathbb{S}^*|\sigma$, if $\neq x = [o]$ and $\neg z = [o]$, then define its lower and upper approximations as in $\mathbb{S}^*$. Define $\pc xz$ in $\mathbb{S}^*|\sigma$ if and only if $\pc xz $ in $\mathbb{S}^*$. }
\end{itemize}
\qed
\end{proof}

\begin{remark}
Clearly this affords a strategy for redefining operations on a partial lu-completion of a pre \textsf{GS} so that it becomes a \textsf{GS}.
\end{remark}

\begin{example}
 
A \textsf{set HGOS} is intended to capture contexts where all objects are described by sets of attributes with related valuations (that is their properties). So objects can be associated with sets of properties (including labels possibly). A more explicit terminology for the concept, may be \emph{power set derived \textsf{HGOS}}(that captures the intent that subsets of the set of all properties are under consideration here). Admittedly, the construction or specification of such a power set is not necessary. In a \textsf{HGOS}, such set of sets of properties need not be the starting point.
 
The difference between a \textsf{HGOS} and a \textsf{set HGOS} at the practical level can be interpreted at different levels of complexity. Suppose that the properties associated with a familiar object like a cast iron frying pan are known to a person $X$, then it is possible to associate a set of properties with valuations that are sufficient to define it. If all objects in the context are definable to a \emph{sufficient level}, then it would be possible for $X$ to associate a \textsf{set HGOS} (provided the required aspects of approximation and order are specifiable). 

It may not be possible to associate a set of properties with the same frying pan in a number of scenarios. For example, another person may simply be able to assign a label to it, and be unsure about its composition or purpose. Still the person may be able to indicate that another frying pan is an approximation of the original frying pan. In this situation, it is more appropriate to regard the labeled frying pan as an element of a \textsf{HGOS}. 

A nominalist position together with a collectivization property can also lead to \textsf{HGOS} that is not a \textsf{set HGOS}.
\end{example}

\begin{definition}
An element $x\in\mathbb{S}$ is said to be \emph{lower definite} (resp. \emph{upper definite}) if and only if $x^l\, =\,x$ (resp. $x^u\, =\,x$) and \emph{definite}, when it is both lower and upper definite. $x\in \mathbb{S}$ is also said to be \emph{weakly upper definite} (resp \emph{weakly definite}) if and only if $ x^u\, =\,x^{uu} $ (resp $ x^u\, =\,x^{uu} \,\&\, x^l =x$ ). 
\end{definition}

Any one of these five concepts may be chosen as a concept of \emph{crispness}. Additional concepts of crispness can be defined through formulas -- it is not necessary that 

In granular operator spaces and generalizations thereof, it is possibly easier to express singletons and the concept of rough membership functions can be generalized to these from a granular perspective. For details see  \cite{am501,am9114}. Every granular operator space can be transformed to a higher granular operator space, but to speak of this in a rigorous way, it is necessary to define related morphisms and categories\cite{am501}.
 
\begin{theorem}
Every \textsf{set HGOS} is a \textsf{HGOS}, every \textsf{HGOS} is a \textsf{GS}, and every \textsf{GS} is a \textsf{GGS}. Further if XX stand for one of \textsf{GGS}, \textsf{GS}, \textsf{HGOS} or \textsf{set HGOS}, then every XX is a Pre XX.
\end{theorem}
\begin{proof}
The proof follows from the definition.  

\end{proof}

\subsection{Rough Objects}\label{roo}

A rough object cannot be known exactly in a rough semantic domain, but can be represented in a number of ways typically through relatively crisp objects.  The following representations of \emph{rough objects} have been either considered in the literature (see \cite{am240,am9006,am501}) or are reasonable concepts that work in the absence of a negation-like operation or relation:

\begin{description}
\item[RL] {$x\in \mathbb{S}$ is a lower rough object if and only if $\neg (x^l\, =\,x) $. }
\item[RU] {$x\in \mathbb{S}$ is a upper rough object if and only if $\neg (x\, =\,x^u) $. }
\item[RW] {$x\in \mathbb{S}$ is a weakly upper rough object if and only if $\neg (x^u\, =\,x^{uu}) $. }
\item[RB] {$x\in \mathbb{S}$ is a rough object if and only if $\neg (x^l\, =\,x^u) $. The condition is equivalent to the boundary being nonempty. }
\item[RD] {\emph{Any pair of definite elements} of the form $(a , b)$ satisfying $a < b $}
\item[RP] {\emph{Any distinct pair of elements} of the form $(x^l ,x^u)$.}
\item[RIA] {Elements in an \emph{interval of the form} $(x^l, x^u)$.}
\item[RI] {Elements in an \emph{interval of the form} $(a, b)$ satisfying $a\leq b$ with $a, b$ being definite elements.}
\item[ET] {In esoteric rough sets \cite{am24}, triples of the form $(x^l, x^{lu}, x^u)$ can be taken as rough objects.} 
\item[RND] {A \emph{non crisp element in a RYS}(see \cite{am240}), that is an $x$ satisfying $\neg \pc x^u x^l   $. This becomes more complicated when multiple approximations are available.}
\end{description}

If a weak negation or complementation $^c$ is available, then orthopairs of the form $(x^l, x^uc)$ can also be taken as representations of \emph{rough objects}.

\subsection{Examples of GGS}

In general, \textsf{GGS} cannot be used to formalize approaches to rough sets that are based on non granular approximations. In fact, this is directly related to the algebraic approach in \cite{am501}, wherein a variant of \textsf{GGS} is used for the granular approach by the present author. \emph{Every rough set approach that relies on granular approximations can be rewritten in terms of \textsf{GGS}}. 
 
A general definition of \emph{point-wise approximations} can be proposed in Second Order Predicate Logic(SOPL) (or alternatively, in a fixed language) based on the following loose SOPL version :  If $S$ is an algebraic system of type $\tau$ and $\nu: S \longmapsto \wp (S)$ is a neighborhood map on the universe $S$, then a \emph{point-wise approximation} $*$ of a subset $X\subseteq S$ is a self-map on $\wp(S)$ that is definable in the form: 
\begin{equation}
X^*\, =\,\{x: \, x\in H\subseteq S \, \& \, \Phi(\nu(x), X) \} 
\end{equation}
for some formula $\Phi(A, B)$ with $A, B \in \wp(S)$. In classical rough sets point-wise approximations lead to a granular semantics, but in other cases they do not in general. 

The full generality implicit in a \textsf{GGS} is not usually required for expressing most granular rough set approaches. So in the following example - this aspect is targeted.

\begin{example}
Suppose the problem at hand is to represent the knowledge of a specialist in automobile engineering and production lines in relation to a database of cars, car parts, calibrated motion videos of cars and performance statistics. The database is known to include a number of experimental car models and some sets of cars have model names, or engines or other crucial characteristics associated. 

Let $\underline{\mathbb{S}}$ be the set of cars, some subsets of cars, sets of internal parts and components of many cars. $\mathcal{G}$ be the set of internal parts and components of many cars. Further let 
\begin{itemize}
\item {$\pc a b$ express the relation that $a$ is a possible component of $b$ or that $a$ belongs to the set of cars indicated by $b$ or that   }
\item {$a \leq b$ indicate that $b$ is a better car than $a$ relative to a certain fixed set of features,}
\item {$a^l$ indicate the closest standard car model whose features are all included in $a$ or set of components that are included in $a$, }
\item {$a^u $ indicate the closest standard car model whose features are all included by $a$ or fusion of set of components that include $a$}
\item {$\vee$, $\wedge$ can be defined as partial operations, while $\bot$ and $\top$ can be specified in terms of attributes. }
\end{itemize}
Under the conditions, \[\mathbb{S} \, =\, \left\langle \underline{\mathbb{S}}, \mathcal{G}, l , u, \pc, \leq , \vee,  \wedge, \bot, \top \right\rangle\] forms a \textsf{GGS}.

Suppose that the specialist has updated her knowledge over time, then this transformation can be expressed with the help of morphisms (see definition \ref{grifggs}) from a \textsf{GGS} to itself. 
\end{example}

\subsection{Extended Example, Fusion}

The difference between fusion ($\mathfrak{F}\subseteq S \times \wp (S)$) and sum ($\Sigma \subseteq S \times \wp (S)$) predicates is relevant in RMCA. Avoiding issues relating to existence, the predicates can be defined as  

\begin{align}
\Sigma a B \stackrel{\vartriangle}{\leftrightarrow} B\subseteq \pc(a) \subseteq \bigcup \{\oc (x): \, x\in B\} \tag{msum}\\
\mathfrak{F} a B \stackrel{\vartriangle}{\leftrightarrow} \oc (a) = \bigcup \{\oc (x): \, x\in B\}   \tag{fusion}
\end{align}

For a set $S$ endowed with a binary parthood relation $\pc$, the set of upper and lower bounds of a subset $X$ are defined by 
\begin{align*}
UB(X) = \{a:\, (\forall x\in X) \pc xa\}   \tag{Upper Bounds}\\
LB(X) = \{a:\, (\forall x\in X) \pc ax\}   \tag{Lower Bounds}
\end{align*}
$S$ is said to be \emph{separative} if and only if SSP (strong supplementation) holds.
\[(\forall a b )(\neg \pc ab \longrightarrow (\exists z) (\pc za \,\&\, \neg \pc zb \,\&\, \neg \pc bz))   \tag{SSP}\]
\begin{theorem}[\cite{rafal2013}]
All of the following hold:
\begin{itemize}
\item {If $\pc$ is reflexive, then  a fusion of $B$ is a mereological sum if it is an upper bound of $B$: \[(\forall a\in S)(\forall B\in \wp (S))(B\subseteq \pc (a) \,\&\, \mathfrak{F} a B \longrightarrow \Sigma a B)\]}
\item {If $\pc$ is transitive and separative then every sum is a fusion and conversely.}
\item {If $\pc$ is transitive and separative then every binary fusion is a binary sum}
\end{itemize}
\end{theorem}

\begin{example}{Fusion and Decisions}\label{progress}
Let $S= \{a,\, b,\, c,\, e, \, f \}$ be a set with parthood $\pc$ defined as the reflexive completion of  
\[ \{(a, c)\, (b, c), \, (a, e),\, (b, e) \}. \]
 
 If $K= \{a,\, b,\, c,\, e\}$, then $\mathfrak{F}cK$ and $\mathfrak{F}eK$ hold. But, $UB(K) = \emptyset$.
 
 Suppose $S$ represents the respective diagnosis of five doctor teams $X$, $W$, $Z$, $E$, and $F$,
 on the basis of diagnostic information indicated in the decision table below. Consider columns $1$, $4$ and $6$ alone first. The sixth column indicates the team type (based on the best performing doctor in the team) involved in the diagnosis. Assume that the doctors are essentially lower approximating an \emph{ideal diagnosis} and that $\pc  \beta \alpha$ means '$\alpha$ is a better diagnosis than $\beta$'.

\begin{table}[hbt]
 \centering
\begin{tabular}{llllllll}
\toprule
\textsf{Doctors} &\textsf{Att:1--3} &\textsf{Att:4--6} &\textsf{Att: 7--9} &\textsf{Diagnosis} & \textsf{Remark} & l & u   \\
\midrule
$X$  &smm & www & nnw & $a$ & General & $X$ & $Z$\\
\midrule
$W$  &mww & swm  &nnn &  $b$ & General & $W$ & $E$ \\
\midrule
$Z$  &smm &mwm  &wmw &   $c$ & Specialist & $Z$ & $Z$\\
\midrule
$E$  &msw  &swm & mms & $e$ & Specialist & $W$ & $E$\\
\midrule
$F$  &mss & mwm &mws &   $f$ & Specialist & $F$ & $F$\\
\end{tabular}
\caption{Doctors and Diagnosis}\label{doci}
\end{table} 
 
 Mereological fusion in the context corresponds to combining expert information. It cannot be used in the context to arrive at any \emph{all encompassing ideal diagnosis}.

The attributes used for the diagnosis are encoded as per: s- severe, m-moderate, w-weak, n-not available. Thus the string \textsf{smm} in the second cell is intended to mean that the valuation for attribute 1 is \textsf{s}, attribute 2 is \textsf{m} and attribute 3 is \textsf{m}. Further suppose that the attributes are potentially causally related, and that the valuations assigned by the doctors are dependent on their own perspectives. The lower and upper approximations of the teams relative to their potential in the context is indicated in the last two columns. 

It is not hard to obtain a granulation based on a simple ordering of the attribute valuation. In practice, the situation is usually more complex. A partial order on $S$ can be specified based on the training of doctor teams and a \textsf{GGS} can be defined on the basis of this information on $S$. It should also be easy to see that the rough inclusion function perspective in the context does not correspond to the approximations. 
\end{example}

\section{Concrete Granular Operator Spaces and Variants}

A \textsf{GGS} need not correspond to a data table (or information system) in general, and from an algebraic perspective it is inconvenient to work with the latter. \emph{Concrete} variants of a \textsf{GGS}, introduced below, include explicit references to objects, and correspond to severe extensions of the idea of data tables. The concepts of \emph{concrete} and \emph{abstract} representation, as used in algebra, are relative in nature and therefore the terminology is justified. The reference to objects is added for the purpose of handling basic concepts used in computations and algorithms.

In data tables, the valuation of object-attribute pairs is in an external relational system. Formulas constructed over resulting triples often determine granulations or approximations. This means that set of   
pairs of the form $(object,\, attributes)$ may be taken as universe (or base) of a partial algebraic system under the assumption that attributes are maps from the set of objects to powerset of possible valuations.
Only in some cases can external valuations and derived relations or correspondences be replaced by formulas that do not refer to the external valuations. It is also a fact that valuations are assumed to be context dependent with few universal features in practice. The same strategy extends to concrete \textsf{GGS} without too many complications. 

One of the standard readings of data tables or information systems can be found in detail in \cite{joanna2008} for example. In \cite{idgg1997,ideo2004}, issues with ordering and morphisms between information systems may be noted from an algebraic perspective. The assumption that the power set of valuations should be a Boolean algebra is often implicit in the considerations. This is significant in the computation and identification of many kinds of reducts whose numbers can be huge (see \cite{js09,zpsk07,sdrskt06}).

\begin{definition}\label{cggs}

A \emph{concrete high general granular operator space} (\textsf{CGGS}) $\mathsf{X}$\\ shall be a two sorted partial algebraic system  of the form \[\mathsf{X}  = \left\langle \underline{\mathbb{O}},\, \underline{\mathbb{S}}, \gamma, \eta, l , u, \pc, \xi , \leq , \vee,  \wedge, \bot, \top \right\rangle\] with $\underline{\mathbb{O}}$, $\underline{\mathbb{S}}$ being sets, \[\left\langle \underline{\mathbb{S}}, \gamma,  l , u, \pc, \leq , \vee,  \wedge, \bot, \top \right\rangle\] being a \textsf{GGS}, $\xi \subseteq \mathbb{O}\times \mathbb{S}$ and $\eta$ being a valuation function \[\eta : \mathbb{O}\times \mathbb{S} \longmapsto V,\] 
with $V$ being a partial algebra of the form                                                                                                                                                                                        
\[V \,=\, \left\langle \underline{V}, \cup, \cap, \sim, 0, 1  \right\rangle \tag{Val-Algebra}\]
with $\underline{V} = \bigcup_a \{ \eta(x, a) : x\in \mathbb{O} \}$ and $\cup, \cap, \sim$ being partial operations satisfying (universal quantifiers are implicit)
\begin{align*}
 a\cap (b \cap c)  \stackrel{\omega}{=} (a\cap b) \cap c) \, \& \, a\cup (b \cup c)  \stackrel{\omega}{=} (a\cup b) \cup c) \tag{WA}\\
a\cup (b \cap c)  \stackrel{\omega}{=} (a\cup b) \cap (a\cup c) \, \& \, a\cap (b \cup c)  \stackrel{\omega}{=} (a\cap b) \cup (a\cap c) \tag{WD}\\
a\cap b  \stackrel{\omega}{=} b \cap a \, \& \, a\cup b  \stackrel{\omega}{=} b \cup a \tag{WC}\\
(a\cap b)\cup a  \stackrel{\omega}{=} a \, \& \, (a\cup b)\cap a  \stackrel{\omega}{=} a \tag{WAb}\\
(\forall a) \, a\cap 0 = a  \& \, a\cup 0  = 0 \,\&\, a\cap 1 = a \,\&\, a\cup 1 = 1\tag{Bo}\\
a\cap \sim a   \stackrel{\omega}{=} 0 \, \& \, a\cup \sim a  \stackrel{\omega}{=} 1\tag{WCp}\\
\sim \sim \sim a  \stackrel{\omega}{=} \sim a \tag{WNeg}
\end{align*}

Further it will be necessary that approximations have been constructed from the valuations through some process $\Phi$.
\end{definition}

\textsf{WA, WD, WC, WAb, WNeg, WCp} respectively are abbreviations for weak associativity, distributivity, commutativity, absorptivity, negation, and complementation respectively in the above. It should be clear that this amounts to presuming a possibly non-Boolean structure on the set of valuations. These may help in avoiding wasting time on useless computations or confusing attribute combinations in practice. \emph{Note that the last line in the definition refers to some process that can be expected to be exactified in specific contexts. This was omitted in an earlier draft of this paper because the context would determine it anyways}.

\begin{definition}
\begin{itemize}
\item {In the definition of a \textsf{CGGS}, if the \textsf{GGS} is replaced by \textsf{Pre-GGS}, then the resulting system would be referred to as a \emph{Pre CGGS}.}
\item {In a \textsf{CGGS}, if the parthood is defined by $\pc ab$ if and only if $a \leq b$ then the \textsf{CGGS} is said to be a \emph{concrete high granular operator space} \textsf{CGS}.}
\item {A \emph{concrete higher granular operator space} (\textsf{CHGOS}) $\mathbb{S}$ is a \textsf{CGS} in which the lattice operations are total.}
\item {In a concrete higher granular operator space, if the lattice operations are set theoretic union and intersection, then the \textsf{CHGOS} will be said to be a \emph{set CHGOS}. }
\end{itemize}
\end{definition}

\begin{example}
In handling datasets relating to structure of humans and cars, an example of a useless conjunction of attributes can be \textsf{has liver} and \textsf{has doors}. 

Similar conjunctions may be avoided in the analysis of data tables (with predefined columns) derived from essays that has comparisons and analogies between humans and cars for fear of confusing readers.
 
\end{example}

\begin{example}{\textsf{Diversity Application Datasets}}\label{diverse}

Suppose datasets have been generated by reviewers from grants/scholarships applications for an event or course from women belonging to diverse backgrounds with diversity and level of systemic discrimination faced being key components of the criteria for recommendation. 

Typically, such applications have subjective components (that may not figure in the columns of the datasets), and recommendations involve additional subjective reasons. Columns in the dataset may concern level of isolation, recent efforts at improvement, academic qualifications, age group, goals and reviewers recommendations. In addition, it may be required of reviewers to avoid comparing applicants 
 
The operations of combining valuations under distinct columns are bound to depend on the applicant - therefore the universality of some conjunctions, disjunctions and  negations may not make sense. The ones that can make sense can be determined through additional analysis by multiple correspondence analysis (MCA) or generalizations thereof. MCA \cite{fsj2011} can be viewed as an                                                                                                                                                                                                                                                                                                                                                                                                                                                                                                                                                                                                                                                                                                                                                                                                                                                                                                                                                                                                                                                                 adaptation of principal component analysis to factor data in which categories are looked for by both columns and rows simultaneously. A detailed 
analysis of such datasets, due to the present author, will appear separately.  
\end{example}

Example \ref{diverse} suggests that additional procedures should be used for\\ identifying operations on valuations.

\begin{theorem}
Let XX stand for one of \textsf{GGS}, \textsf{GS}, \textsf{HGOS} or \textsf{set HGOS}, then every concrete XX has XX as a projection. Also every concrete \textsf{set HGOS} is transformable into an information system.
\end{theorem}

\begin{proof}
For the second part, note that the set of objects and attributes are respectively constructible from a \textsf{set HGOS}. The valuation function can also be extracted. Information about approximations and granulations cannot be transformed in any obvious manner.

\qed
\end{proof}

\subsection{Improved Inverse Problem}

The inverse problem is a partly broken representation problem in the context of general rough sets due to the present author. An overview can be found in \cite{am5019}. The basic problem is \textsf{to find an information table (system) or a set of approximation spaces that fits the available information in the form of a set of approximations, similarities, and some relations about objects in accordance with a rough procedure}.
In the presented form, it can be very difficult to solve. In \cite{am5019}, it is mentioned by the present author granular operator spaces and higher order variants can be used for its formulation. This formulation can be severely improved using the concepts of CGGS, Pre-CGGS  and special morphisms. Note that even CHGOS need not correspond to information systems.

\begin{definition}
Let $\mathsf{X} \, =\, \left\langle \underline{\mathbb{O}},\, \underline{\mathbb{H}}, \gamma, \eta, l , u, \pc, \leq , \vee,  \wedge, \bot, \top \right\rangle$ be a Pre-CGGS. By a \emph{projective closed morphism} of a Pre-GGS $\mathbb{S}$ into a Pre-CGGS $\mathsf{X}$ will be meant a map $\varphi: \mathbb{S} \longmapsto \mathbb{H}$ that satisfies all of the following:
\begin{itemize}
\item {$\varphi$ is a closed morphism and}
\item {is injective}
\end{itemize}.
\end{definition}

\begin{problem}
The inverse problem(improved) for a Pre-GGS is then the problem of construction of a projective, closed morphism into a Pre-CGGS.  Analogously, the problem can be specialized to Pre-GS, Pre-HGOS and Set Pre-HGOS. 
\end{problem}

Possible solutions of the problem depend on granulations and approximations used, and some solutions are known \cite{am501,bc1} in the following sense. 

\begin{theorem}
For every pre-rough algebra $S$, there exists an approximation space $X$ such that the pre-rough algebra generated by $X$ is isomorphic to $S$.
\end{theorem}
\begin{theorem}
For every super rough algebra $S$, there exists an approximation space $X$ such that the super rough set algebra generated by $X$ is isomorphic to $S$.
\end{theorem}

But the above representation results are about approximation spaces and not information tables. A number of information tables (and therefore Pre-GGS) can give rise to the same approximation space.

In the context of this paper, \emph{solutions of the inverse problem can directly help in defining appropriate granular rough inclusion functions on GGS from among the many that would be possible in general} -- this aspect is transformed into a strategy for partially solving the problem in the seventh section.

\subsection{Less-Contaminated Numeric Measures}\label{numm}

The concept of a \emph{careful measure} introduced below is intended to express a numeric measure that depends on definite or crisp objects alone. \emph{The numeric part can still be a source of contamination and noise unless such objects are essentially made of granules of the same weight}.

\begin{definition}
Let $\mathbb{S}$ be a Pre-GGS.  A partial function $f: \mathbb{S}^n \longmapsto \Re$ will be said to be a \emph{careful measure} if and only if $dom(f) \subseteq (R(\mathcal{S}))^n$ for a positive integer $n$, where 
 $R(\mathcal{S}$ is a set of definite or crisp objects of $\mathbb{S}$. 
\end{definition}

To accommodate Pre-CGGS, an additional map like the one defined next may be used.

\begin{definition}
Let $\mathsf{S}$ be a Pre-CGGS.  A function $\varphi : \mathbb{S}^r \longmapsto R(\mathsf{S})^n$ will be said to be an \emph{approximator} if $r, n$ are positive integers and $R(\mathsf{S})^n$ is a set of definite or crisp objects of $\mathsf{S}$.
\end{definition}

In this research paper, allowance will be made for composition of these two classes of functions alone.

\subsection{Reducts}

Explicit reference to objects in a CGGS and variants permits one to easily speak of indiscernibilities and reducts of various types.

Let $\Phi_{\rho} (a, b, x)$ be some formula that says that attribute set $x$ can discern object $a$ from $b$ in the sense $\rho$. 

\begin{definition}
Given a \textsf{CGGS} $\mathsf{X}$ of the form \[ \left\langle \underline{\mathbb{O}},\, \underline{\mathbb{S}}, \gamma, \eta, l , u, \pc, \leq , \vee,  \wedge, \bot, \top \right\rangle,\] and discernibility formula of the form  $\Phi_{\rho} (a, b, x)$, a \emph{skewed discernibility matrix} will be a square matrix $\Delta$ of order $\#(\mathbb{O})$ with its entries defined by ($a_i $ being the $ith$ object in $\mathbb{O}$) 
\[\delta_{ij} = \{x; \, x\in \mathbb{S}\, \&\, \Phi_{\delta} (a_i, a_j, x)\}\] 
\end{definition}

The entries in a skewed discernibility matrix are most likely to contain superfluous elements. In all cases, each entry should be replaced by a $\pc$-minimal subset (even if $\pc$ is not even transitive). These $\pc$-minimal matrices will be referred to as \emph{$\pc$-minimal skewed discernibility matrices}.

\begin{proposition}
In the context of set CHGOS, $\pc$-minimal skewed discernibility matrices are proper generalizations of the classical concept of a discernibility matrix. 
\end{proposition}

\begin{proof}
Since $\pc$ is the same as set inclusion in the context, $\pc$-minimal elements in the $(ij)$'th position are minimal subsets of attributes that distinguish between the $i$th and $j$th objects. In the classical way of computing discernibility matrices, if the $ij$ th entry of the matrix is $\{1, 2, 3 \}$, then the corresponding entry in the $\pc$-minimal skewed discernibility matrix would be $\{\{1\},\{2\}, \{3\} \}$.
So $\pc$-minimal skewed discernibility matrices are proper generalizations. 
\qed
\end{proof}

The reader may refer to \cite{am9204,jc1977} for more on the following definition.
\begin{definition}\label{fra}
Let the \emph{principal up-set} generated by $a, b\in \mathbb{S}$ be the set \[U(a, b) = \{x: \pc ax \,\&\,\pc b x \}.\] 

$K\subset \mathbb{S}$ is a $\pc$-\emph{ideal} if and only if 
\begin{align}
(\forall x\in \mathbb{S})(\forall a\in K)(\pc xa \longrightarrow x\in K)\\
(\forall a, b\in K)\, U(a, b) \cap K \neq \emptyset 
\end{align}

If a $\pc$-ideal $K$ is representable as the intersection of all ideals containing a single element $a\in \mathbb{S}$, then it is said to be \emph{principal}.
\end{definition}

Concepts of skewed information reducts can be directly defined relative to a $\pc$-minimal skewed discernibility matrix as principal $\pc$- ideals of $\mathbb{S}$ that preserve the $\pc$-minimal skewed discernibility matrix. 

 Because reducts are not used in this paper, these will be taken up separately.

\section{General Rough Inclusion Functions}

Intuitively, generalizations of rough inclusion functions are likely to work perfectly when 
\begin{description}
\item[A1]{the contribution of attributes to approximations have uniform weightage across approximations,}
\item[A2]{the contributions of attributes in the construction of approximation can be assigned weights,}
\item[A3]{the functions are robust (that is the value of the function does not change much with small deviations of its arguments \cite{skajsd2016}) and stable relative to the context,}
\item[A4]{Every aggregate of attributes is meaningful, and}
\item[A5]{Attributes are independent.}
\end{description}
The ideas of robustness and stability are always relative to a finite number of purposes or use cases in application contexts.

In this section, the different known rough inclusion functions are generalized to \textsf{GS} of the form $\mathbb{S} \, =\, \left\langle \underline{\mathbb{S}}, \mathcal{G}, l , u, \pc, \vee,  \wedge, \bot, \top \right\rangle$. If $\kap : \underline{\mathbb{S}}^2 \longmapsto [0, 1]$ is a map, consider the conditions,
\begin{align*}
(\forall a)\, \kap (a, a) = 1    \tag{U1}\\
(\forall a, b)(\kap (a, b) = 1 \leftrightarrow   \pc a b) \tag{R1}\\
(\forall a, b, c)(\kap (b, c) = 1 \longrightarrow \kap (a, b) \leq \kap (a, c))   \tag{R2}\\
(\forall a, b, c)(\pc bc \longrightarrow \kap (a, b) \leq \kap (a, c))   \tag{R3}\\
(\forall a, b)(\pc ab \longrightarrow \kap (a, b) = 1)   \tag{R0}\\ 
(\forall a, b)(\kap (a, b) = 1\longrightarrow \pc ab )   \tag{IR0}\\ 
(\forall a) (\pp \bot a\longrightarrow  \kap (a, \bot ) = 0)   \tag{RB}\\
(\forall a, b)(\kap ( a, b)=0 \longrightarrow a \wedge b = \bot)  \tag{R4}\\
(\forall a, b)( a\wedge b = \bot \,\&\, \pp \bot a \longrightarrow \kap (a, b) = 0)   \tag{IR4}\\ 
(\forall a, b)(\kap ( a, b)=0 \,\&\, \pp \bot a  \leftrightarrow a \wedge b = \bot)  \tag{R5}\\
(\forall a, b, c)( \pp \bot a \, \&\, b\vee c = \top\longrightarrow  \kap (a , b) + \kap (a, c) =1)  \tag{R6}
\end{align*}

These mostly correspond to the definition in \cite{ag2009,ag3}. $rif_3$ is RB, and $rif_{2*}$ is R2 under the conditions mentioned. Proofs of the next proposition can be found in the \cite{ag2009}. These carry over to \textsf{HGOS} directly, while the proofs in a \textsf{GS} are not hard.

\begin{theorem}
The following implications between the properties are easy to verify.
\begin{description}
\item[prif1]{If a \textsf{GS} $\mathbb{S}$ satisfies R1, then R3 and R2 are equivalent.}
\item[prif2]{R1 if and only if R0 and IR0 are satisfied.}
\item[prif3]{R0 and R2 imply R3.}
\item[prif4]{IR0 and R3 imply R2.}
\item[prif5]{IR4 implies RB.}
\item[prif6]{IR4 and R4 if and only if R5.}
\item[prif7]{When complementation is well defined then R0 and R6 imply IR4.}
\item[prif8]{When complementation is well defined then IR0 and R6 imply R4.}
\item[prif9]{When complementation is well defined then R1 and R6 imply R5.}
\end{description}
Further both R1 and R0 imply U1.
\end{theorem}
\begin{proof}
Aspects of the proof are illustrated below
\begin{description}
\item[prif1]{Suppose $\pc bc$ for some $b, c\in \mathbb{S}$, then by R1 $\kap (b, c) =1$ and conversely. In R2 and R3 the premise can be interchanged in the conditional implication when R1 holds.}
\item[prif2]{Obvious.}
\item[prif3]{Suppose R0 and R2 hold. If $\pc bc$ for some $b, c\in \mathbb{S}$ then by R0 $\kap (b, c) =1$. So for any $a$ $\kap (a, b) \leq \kap (a, c)$. That is R3 follows from R0 and R2}
\item[prif5]{Substituting $\bot$ for $b$ in IR4 yields RB.}
\item[pref6]{Is obvious.}
\end{description}

\end{proof}

\begin{definition}
By a \emph{general rough inclusion} function (RIF) on a \textsf{GS} $\mathbb{S}$ shall be a map $\kap : (\mathbb{S})^2 \longmapsto [0, 1]$ that satisfies R1 and R2.
A \emph{general quasi rough inclusion} function (qRIF) will be a map $\kap : (\mathbb{S})^2 \longmapsto [0, 1]$ that satisfies R0 and R2. While a \emph{general weak quasi rough inclusion} function (wqRIF) will be a map $\kap : (\mathbb{S})^2 \longmapsto [0, 1]$ that satisfies R0 and R3. 
\end{definition}

\begin{proposition}
In a \textsf{GS} $\mathbb{S}$, every \textsf{RIF} is a \textsf{qRIF} and every \textsf{qRIF} is a \textsf{wqRIF}.  
\end{proposition}

\subsection{Specific Weak Quasi-RIFs}

RIFs and variants thereof are defined over power sets in \cite{ss2010,ag2009}. For rewriting them in the high granular operator space way, it is necessary to assume that $\underline{\mathbb{S}} = \wp(\top)$, $\top$ being a finite set, $\bot = \emptyset$,  $\pc = \leq = \subseteq$, $\vee = \cup$ and $\wedge = \cap$. Specifically, the following functions have been studied in \cite{ag2009} and have been used to define concepts of approximation spaces.   

\begin{equation*}\label{k1}
\nu_1(A, B) = \left\lbrace  \begin{array}{ll}
 \dfrac{\# (B)}{\# (A\cup B)} & \text{if } A\cup B\neq \emptyset\\
 1 & \text{otherwise}\\
 \end{array} \right. \tag{K1}                                                                                                             
\end{equation*}

\begin{equation*}\label{k2}
\nu_2(A, B) =  \dfrac{\# (A^c \cup B)}{\# (\top)}  \tag{K2}                                                                                                             
\end{equation*}

If $0 \leq s < t \leq 1$, and $\nu : {\mathbb{S}}^2 \longmapsto [0, 1]$ is a \textsf{RIF}, then let $\nu_{s,t}^{\nu} : {\mathbb{S}}^2 \longmapsto [0, 1]$ be a function defined by 

\begin{equation*}\label{kst}
\nu^{\nu}_{s,t}(A, B) = \left\lbrace  \begin{array}{ll}
0 & \text{ if } \nu(A, B) \leq s \\
 \dfrac{\nu (A. B) - s}{t - s} & \text{ if } s < \nu(A, B) < t,\\
 1 & \text{ if } \nu(A, B ) \geq t\\
 \end{array} \right. \tag{Kst}                                                                                                             
\end{equation*}

\begin{proposition}{\cite{ag2009}} 
In general, $\nu^{\nu}_{s,t}$ is a weak quasi \textsf{RIF} and $\nu^{\nu}_{s,1}$ is a quasi \textsf{RIF}. 
\end{proposition}

\subsection{Operations on Generalized \textsf{RIF}s}

Operations on these generalized \textsf{RIF}s (including \textsf{RIF}s) have not been studied in the literature on rough sets to date - though there are interesting results on the set of all rough membership functions \cite{mkc2014} and on generalized rough membership functions in granular operator spaces \cite{am9114}(that can be read as a generalized rough inclusion partial functions). The operations on the set of membership function are based on point-wise ordering in \cite{mkc2014} and on two modal operators that essentially transform \emph{the rough membership function in a set $A$} by \emph{the rough membership function in the set $A^l$ and $A^u$} respectively. The modal operators reduce contamination to an extent (the author does not say as much in the paper).  This also suggests that \textsf{RIF}s and generalizations thereof should be handled from the perspective of order, and generalized disjunctions and conjunctions. This motivates the use of t-norms and s-norms -- it should be noted that the result of operations may not have a concrete interpretation associated in the rough set theoretical scheme of things. For example, for four elements $A, B, C, E $ of a \textsf{GS} $\mathbb{S}$, $\nu_1 (A, B) \otimes \nu_1(C, E)$ may be computable given the  choice of the definition of $\otimes$, but the result may not make sense. If granularity is somehow involved, then it may be possible to be more sure of the outcome. 
  
At the same time it is possible to assign some meaning to addition and product (in the field of real numbers) of generalized \textsf{RIF}s, though its practical value is limited. 

\subsection{Connections of Generalized RIFs with Other Measures}

RIFs have been related to a number of numeric measures such as quality of classification\cite{zpb} , variable precision rough sets \cite{zw,ss1,yec2017}, accuracy degree of approximation \cite{zpb}, degrees of closeness \cite{lp2011}, dependence degree of a set of attributes on another \cite{zpb}, dependency degree of a decision set with respect to an attribute set and others. Rough membership functions are usually not related to rough inclusion functions in the literature. In most of these cases, the RIFs involve possibly non-crisp and non-definite objects. Some of the connections are mentioned in \cite{jzd2003}. 

Variants of RIFs have also been used in reduct computation for contexts involving non granular rough approximations (see \cite{chen2014,wu2002,skg1994}). Limitations of mass and plausibility functions are also mentioned in the context.

\begin{theorem}
If $\mathbb{S}$ is a \textsf{set HGOS}, then 
\begin{itemize}
\item {The accuracy degree of approximation of an element $x$ is \[\alpha(x)  = \dfrac{\#(x^l)}{\#(x^u)} = \nu(x^u, x^l).\] }
\item {The classical rough inclusion degree $\nu(a, b)$ defined by Equation. \ref{rif0} is not a function of crisp objects. The degree of misclassification is $\mu(a, b) = 1 - \nu(a, b)$. It coincides with $\nu(a, b^c)$ whenever $\mathbb{S}$ is closed under complements. }
\item {Relative to a partition $\mathcal{S}$ satisfying $\bigcup \mathcal{S} = \bigcup \mathbb{S}$, or even relative to the granulation $\mathcal{G}$, it is possible to define generalized VPRS approximations of any $X\in \mathbb{S}$ for a pair of parameters $0 < \alpha \leq \beta < 1$ as follows:
\begin{align*}
X^{l_v} = \bigcup \{h: \, h\in \mathcal{G}\, \&\, \nu(h, X) > \beta\} \\
X^{u_v} = \bigcup \{h: \, h\in \mathcal{G}\, \&\, \nu(h, X) > \alpha\}
\end{align*}
These approximations clearly depend on granules and the original set.}
\end{itemize}
\end{theorem}

The latter definition is fixed next:

\begin{definition}
Relative to the granulation $\mathcal{G}$, it is possible to define \emph{fixed generalized VPRS approximations} of any $X\in \mathbb{S}$ for a pair of parameters $0 < \alpha \leq \beta < 1$ as follows:
\begin{align*}
X^{l_v} = \bigcup \{h: \, h\in \mathcal{G}\, \&\, \nu(h, X^l) > \beta\} \\
X^{u_v} = \bigcup \{h: \, h\in \mathcal{G}\, \&\, \nu(h, X^l) > \alpha\}
\end{align*}
These approximations clearly depend on granules or approximations. 
\end{definition}

In the context of \textsf{set HGOS}, if GRIFs are used then they can be shown to generate few generalized or very different relationships. While it is possible to define a number of generalizations on the theme, the most interesting ones are those of relative comparison between different pairs of approximations.

\section{Granular Rough Inclusion Functions}\label{granrif}

To capture rough reasoning, it is important to avoid using objects and operations that are actually accessible only in an exact perspective of the context. \emph{Every concept of a weak quasi rough inclusion function considered in the previous section is flawed relative to this perspective}. 

The concept of contamination relates to the contexts under consideration and many levels of contamination reduction can be of interest in practice. If all objects permitted in a domain are approximations of objects in the classical semantic domain, then the domain in question would be referred to as a \emph{weak rough domain}. One way of reducing contamination can be through the use of approximations or representations of rough objects instead of sets that are not perceived as such in a weaker rough or rough semantic domain respectively. 
But this can be done in a number of ways because an object of the classical domain can be expressed by a number of approximations in rough domains. \emph{Therefore, it is reasonable to assume that measures include all or most important possibilities}.

\emph{It is also of interest to use weights corresponding to different attributes or at least granules guided by additional rules. The behavior of aggregations and commonality operations used can be useful in constructing algorithms for the same}.

Over a \textsf{set HGOS}, the following conceptual variants of rough inclusion functions can be defined. 

\begin{definition}\label{grifhgos}
If $\mathbb{S}$ is a \textsf{set HGOS}, $A, B\in \mathbb{S}$, $\sigma, \pi \in \{l, u\}$  and the denominators in the expression is non zero, then let 
\begin{equation*}
\nu_{\sigma \pi}(A, B ) = \dfrac{\#(A^\sigma \cap B^\pi )}{\#(A^\sigma )}  \tag{$\sigma\pi$-grif1} 
\end{equation*}
If $\# (A^\sigma) = 0 $, then set the value of $\nu_{\sigma \pi}(A, B)$ to $1$. 
$\nu_{\sigma \pi}$ will be said to be a \emph{basic granular rough inclusion function}.
\end{definition}

Variants of this definition are also of interest:

\begin{definition}\label{cogrifhgos}
If $\mathbb{S}$ is a \textsf{set HGOS}, $A, B\in \mathbb{S}$, $\sigma, \pi \in \{l, u\}$  and the denominators in the expression is non zero, then let 
\begin{equation*}
\nu_{\sigma \pi}(A, B ) = \dfrac{\#(A^\sigma \cap B^\pi )}{\#(A^\pi )}  \tag{$\sigma\pi$-grif2} 
\end{equation*}
If $\# (A^\pi) = 0 $, then set the value of $\nu_{\sigma \pi}(A, B)$ to $1$. 
$\nu_{\sigma \pi}$ will be said to be a \emph{cobasic granular rough inclusion function}.
\end{definition}

\begin{proposition}
In the context of Definition \ref{grifhgos}, \[(\forall A, B\in \mathbb{S} ) \,0\leq  \nu_{\sigma \pi} (A, B) \leq 1\] This proposition does not hold for cobasic granular rough inclusion functions in general.
\end{proposition}

\begin{theorem}
In a \textsf{set HGOS} $\mathbb{S}$ with $\bot=\emptyset$, all of the following hold ($\alpha$ being any one of $ll, lu, ul$ or $uu$):
\begin{align*}
(\forall A, B)\, \nu_{ul}(A, B ) \leq \nu_{uu}(A, B )  \tag{ulu2}\\
(\forall A, B)\, \nu_{ll}(A, B ) \leq \nu_{lu}(A, B )  \tag{llu2}\\ 
(\forall A, B, E)\,(B\subset E \longrightarrow \, \nu_{\alpha}(A, B) \leq \nu_{\alpha}(A, E) )   \tag{mo}\\
(\forall A )\,\nu_{lu}(A, A) \leq \nu_{ll}(A, A) = 1 = \nu_{uu}(A, A)    \tag{refl}\\
(\forall A)\,\nu_{\alpha}(\bot , A) =1   \tag{bot}\\
(\forall A)\, (\top=\top^l = \top^u \longrightarrow \nu_{\alpha}(A,\top) =1)   \tag{top}
\end{align*}
\end{theorem}
\begin{proof}
 \begin{itemize}
\item {\textsf{ulu2} follows from $(\forall B \in \mathbb{S})\, B^l \subseteq B^u$.  }
\item {Proof of \textsf{llu2} is similar.}
\item {Since both $l$ and $u$ are monotonic and the denominator is invariant in\\ $\nu_{\alpha}(A, B) \leq \nu_{\alpha}(E, B)$, \textsf{mo} follows. }
\item {Proof of \textsf{refl} is direct.}
\item {Only in the condition $\bot$, is the assumption $\bot = \emptyset$ used.}
\end{itemize}
\qed 
\end{proof}

All this leads to

\begin{definition}\label{griff1}
In a \textsf{set HGOS} $\mathbb{S}$, by the \emph{granular rough inclusion function of type-1} (GRIF-1) $\zeta^{\nu}$ will be meant a function $: \mathbb{S}^2 \longmapsto M_Q$ ($M_Q$ being the set of $2\times 2$ matrices over the set $Q_+\cap [0, 1]$, of positive rationals less than or equal to $1$) defined for any $(A, B)\in \mathbb{S}^2$ as below:
\[
\zeta^{\nu} (A, B) = 
\begin{pmatrix}
\nu_{ll}(A, B ) & \nu_{lu}(A, B ) \\
\nu_{ul}(A, B ) & \nu_{uu}(A, B )
\end{pmatrix}  
\]

In the context of Definition \ref{griff1}, if $Q_+ \cap [0,1]$ is replaced by the real unit interval $[0, 1]$,  and $M_Q$ by $M_I$ (set of $2\times 2$ matrices over the real unit interval $[0,1]$), then the function will be called a \emph{granular rough inclusion function of type-0} (GRIF-0).
\end{definition}
\begin{remark}
The most appropriate domain for generalized \textsf{RIF}s depends on the application context and the philosophical assumptions made. Therefore finite subsets of $Q_+ \cap [0,1]$, $Q_+ \cap [0,1]$, and $[0,1]$ can all be relevant. 

Because general \textsf{RIF}s do not have any ontological commitment to the cardinality of sets, GRIF-1 is not the only possibility.
\end{remark}

The set of $2\times 2$ matrices over the field of rationals forms a noncommutative ring, but a direct interpretation of the ring operations is not possible in the context. Denoting an arbitrary t-norm by $\otimes$ on the set $Q_+ \cap [0, 1]$ and another arbitrary s-norm by $\oplus$, the following operations can be defined for any $(a_{ij}),\,(b_ij)\in M_Q$: 

\begin{align*}
(a_{ij})\varovee (b_{ij}) := (a_{ij} \oplus b_{ij})  \tag{disjunction}\\
(a_{ij}) \varowedge (b_{ij}) := (\bigoplus_{k} a_{ik}\otimes b_{kj})   \tag{conjunction}
\end{align*}

More generally,
\begin{definition}
In a high granular operator space $\mathbb{S}$, if $\tau$ is a \textsf{wqRIF}, 
then the \emph{granular weak quasi rough inclusion function} $\zeta^{\tau}$ (GwqRIF) induced by $\tau$ will be a 
meant a function $: \mathbb{S}^2 \longmapsto M_Q$ defined for any $(A, B)\in \mathbb{S}^2$ as below:
\[
\zeta^{\tau} (A, B) = 
\begin{pmatrix}
\tau(A^l, B^l ) & \tau(A^l, B^u ) \\
\tau(A^u, B^l ) & \tau(A^u, B^u )
\end{pmatrix}  
\] 
\end{definition}

In general, these definitions do not necessarily involve definite objects. 

\begin{theorem}
 \begin{itemize}
\item {In a high granular operator space $\mathbb{S}$ if $\tau$ is a weak quasi \textsf{RIF} and $A$ is a definite element then
\begin{equation}
\zeta^{\tau}(A, B) = \begin{pmatrix}
\tau(A, B^l ) & \tau(A, B^u) \\
\tau(A, B^l ) & \tau(A, B^u)
\end{pmatrix}  
\end{equation}}
\item {If $A^l = A=A^u$ and $B^l = B = B^u$, then  
\begin{equation}
\zeta^{\tau} (A, B) = \begin{pmatrix}
\tau(A, B ) & \tau(A, B ) \\
\tau(A, B ) & \tau(A, B )
\end{pmatrix}
\end{equation} }
\item {If $B^l = B = B^u$, then  
\begin{equation}
\zeta^{\tau} (A, B) = \begin{pmatrix}
\tau(A^l, B ) & \tau(A^l, B ) \\
\tau(A^u, B ) & \tau(A^u, B )
\end{pmatrix}
\end{equation} }
\end{itemize}
\end{theorem}

The above theorem motivates the following definition:
\begin{definition}
In a high granular operator space $\mathbb{S}$ if $\tau$ is a weak quasi \textsf{RIF} and
\begin{itemize}
\item {if $A$ is a definite element, then let $\xi^{\tau}(A, B) = (\tau(A, B^l ) , \tau(A, B^u))$ and}
\item {if $B$ is a definite element, then let $\omega^{\tau}(A, B)  = (\tau(A^l, B ), \tau(A^u, B ))$}
\end{itemize}
$\xi$ and $\omega$ will respectively be referred to as the \emph{1-certain} \textsf{GRIF} (1GwqRIF) and \emph{2-certain} \textsf{GRIF} (2GwqRIF) induced by the weak quasi \textsf{RIF} $\tau$.  
\end{definition}

\emph{This suggests that \textsf{GwqRIF}s may be viewed as semilinear transformations of\\ 2GwqRIFs and 1GwqRIFs}. Clearly there is much to be fixed for this view that \emph{pairs of inclusion measures of objects in crisp objects correspond to inclusion measures of objects in other not necessarily crisp objects} or that \emph{pairs of inclusion measures of crisp objects in not necessarily crisp objects correspond to inclusion measures of objects in other objects}.  

\begin{theorem}
$M_Q$ along with the operations $\varowedge$, $\varovee$ and neutral elements, $\mathbf{0}$ and $\mathbf{1}$ forms a semiring with unity when $\oplus$ is the Min t-norm operation.

Dually, $M_Q$ along with the operations $\varowedge$, $\varovee$ and neutral elements, $\mathbf{0}$ and $\mathbf{1}$ forms a dual semiring with unity when $\otimes$ is the Max s-norm operation.
\end{theorem}
\begin{proof}
\begin{itemize}
\item { $\otimes$ distributes over $\oplus$ if and only if $\oplus$ is the Max s-norm (see Theorem 3.5.1 in \cite{cjb2006}.}
\item {For any three elements $(a_{ij})$, $((b_{ij})$ and $(c_{ij}))$ in $M_Q$,  
\begin{align*}
(a_{ij}) \varowedge ((b_{ij}) \varovee (c_{ij})) = (a_{ij}) \varowedge (b_{ij} \oplus c_{ij}) =\\ 
(\bigoplus_{k}[a_{ik} \otimes (b_{kj} \oplus c_{kj})]) = (\bigoplus_{k}(a_{ik} \otimes b_{kj}) \oplus (a_{ik} \otimes c_{kj}))\\ 
\text{the last step holds whenever distribution of $\otimes$ over $\oplus$ holds.} 
\end{align*}}
\end{itemize}

The dual version of this holds because $\oplus$ distributes over $\otimes$ if and only if $\otimes$ is the Min t-norm.
\qed 
\end{proof}

\paragraph{Discussion}
Clearly this shows that sets of \textsf{GRIF}s have a nice algebraic structure associated in a number of situations. Distributivity is not really essential for the purpose of this paper. In the absence of the property, it is possible to use the H-product instead mentioned in the Sec. \ref{bck} to ensure it. But it is a fact that it is not the best of operations from the point of view of related morphisms.    

The theorem provides yet another way of deciding on when a fuzzy strategy can possibly mimic a granular rough set approach in a transparent way because
\begin{itemize}
\item {representation of \textsf{GwqRIF}s through 2GwqRIFs need not hold in general, }
\item {but such representation can possibly be approximated through\\ choice of t-norms and s-norms,}
\item {choice of t-norms and s-norms correspond to a fuzzy reasoning strategy, and}
\item {all this is reasonable from a general inclusion function perspective.}
\end{itemize}

\subsection{Results on Form of Matrix}

\begin{theorem}
In a \textsf{set HGOS} $\mathbb{S}$, if $\tau(A, B)$ is the standard rough inclusion function defined by Equation..\ref{rif0}, and $A^l \neq \emptyset$, then \begin{equation}
\zeta^{\tau} (A, B) = \begin{pmatrix}
1 & 1 \\
r & 1
\end{pmatrix}
\end{equation} for a $r \leq 1$ if and only if \[A^l \subset B^l \subset A^u \subset B^u \text{ or } A\subset B \text{ or } A = B \]
\end{theorem}

\begin{proof}

If $A = B$, then \begin{equation}
\zeta^{\tau} (A, A) = \begin{pmatrix}
1 & 1 \\
\dfrac{\#(A^l)}{\#(A^u)} & 1
\end{pmatrix}
\end{equation}
and $r = \dfrac{\#(A^l)}{\#(A^u)} \leq 1$

If $A\subset B$, then $A^l \subseteq B^l$ and $A^u \subseteq B^u$.
So \begin{equation}
\zeta^{\tau} (A, B) = \begin{pmatrix}
1 & 1 \\
\dfrac{\#(A^u\cap B^l)}{\#(A^u)} & 1
\end{pmatrix}
\end{equation}
and $r = \dfrac{\#(A^u\cap B^l)}{\#(A^u)} \leq 1$

If $A^l \subset B^l \subset A^u \subset B^u$, then it is not necessary that $A\subseteq B$ and 
\begin{equation}
\zeta^{\tau} (A, B) = \begin{pmatrix}
1 & 1 \\
\dfrac{\#(A^u\cap B^l)}{\#(A^u)} & 1
\end{pmatrix}
\end{equation}
and $r = \dfrac{\#(A^u\cap B^l)}{\#(A^u)} < 1$

For the converse, suppose that \begin{equation}
\zeta^{\tau} (A, B) = \begin{pmatrix}
1 & 1 \\
r & 1
\end{pmatrix}
\end{equation}
for a $r \leq 1$, then it is necessary that 

\begin{align*}
\#(A^l\cap B^l) = \#(A^l) \text{ or } A^l  = \emptyset  \tag{llc}\\
\#(A^l\cap B^u) = \#(A^l) \text{ or } A^l  = \emptyset  \tag{luc}\\
\#(A^u\cap B^l) \leq \#(A^u) \text{ or } A^u  = \emptyset \tag{ulc}\\
\#(A^u\cap B^u) = \#(A^u)  \text{ or } A^u  = \emptyset \tag{uuc} 
\end{align*}

Because, the sets are finite, it follows that 
\begin{align}
A^l \subseteq B^l   \tag{ll1}\\
A^l \subseteq B^u   \tag{lu1}\\
A^u \nsubseteq B^l   \tag{ul1}\\
A^u \subseteq B^u   \tag{uu1}
\end{align}

These four conditions are equivalent to \textsf{ll1, ul1, uu1}, which is possible only when
\[A^l \subset B^l \subset A^u \subset B^u \text{ or } A\subset B \text{ or } A = B \]
 
\qed 
\end{proof}

\begin{theorem}
In a \textsf{set HGOS} $\mathbb{S}$, if $\tau(A, B)$ is the standard rough inclusion function defined by Equation..\ref{rif0},  then \begin{equation}
\zeta^{\tau} (A, B) = \begin{pmatrix}
0 & 1 \\
1 & 1
\end{pmatrix}
\end{equation} is not possible
\end{theorem}

\begin{proof}

Suppose, \begin{equation}
\zeta^{\tau} (A, B) = \begin{pmatrix}
0 & 1 \\
1 & 1
\end{pmatrix}
\end{equation}

Then it is necessary that
\begin{align}
A^l \cap B^l = \emptyset \, \& \, A^l \neq \emptyset   \tag{ll0}\\
A^l \subseteq B^u \tag{lu0}\\
A^u \subseteq B^l   \tag{ul0}\\
A^u \subseteq B^u   \tag{uu0}
\end{align}

But this is impossible.
\qed
\end{proof}

\begin{theorem}
In a \textsf{set HGOS} $\mathbb{S}$, if $\tau(A, B)$ is the rough inclusion function defined by equation \ref{k1},
\begin{equation*}
\tau(A, B) = \left\lbrace  \begin{array}{ll}
 \dfrac{\# (B)}{\# (A\cup B)} & \text{if } A\cup B\neq \emptyset\\
 1 & \text{otherwise}\\
 \end{array} \right. \tag{K1}                                                                                                             
\end{equation*}
then 
\begin{equation}
\zeta^{\tau} (A, B) = \begin{pmatrix}
1 & 1 \\
1 & 1
\end{pmatrix}
\end{equation} 
if and only if one of the following three conditions holds
\begin{align}
A^u \cup B^u = \emptyset   \tag{1}\\
A^u \cup B^l = \emptyset \, \&\, B^u \neq \emptyset   \tag{2}\\
A^l \subseteq A^u \subseteq B^l \subseteq B^u   \tag{3}
\end{align}
\end{theorem}
\begin{proof}
If $A^u \cup B^u = \emptyset$, then the form of $\zeta_{\tau}(A, B)$ is obvious.

If $A^u \cup B^l = \emptyset \, \&\, B^u \neq \emptyset$, then $A^l \cup B^l = \emptyset$, $A^l = \emptyset$, and $A^u = \emptyset$. So the value of  $\zeta_{\tau}(A, B)$ follows.
 
If $A^l \subseteq A^u \subseteq B^l \subseteq B^u $, then the denominator in elements of the first and second column respectively of $\zeta_{\tau}(A, B)$ are $\#(B^l)$ and $\#(B^u)$ respectively.

For the converse, suppose that 
\begin{equation}
\zeta^{\tau} (A, B) = \begin{pmatrix}
1 & 1 \\
1 & 1
\end{pmatrix}
\end{equation} 
Then it is necessary that 
\begin{align}
\#(B^l) = \#(A^l \cup B^l)  \text{ or } A^l \cup B^l = \emptyset  \tag{b1}\\
\#(B^u) = \#(A^l \cup B^u)  \text{ or } A^l \cup B^u = \emptyset \tag{b2}\\
\#(B^l) = \#(A^u \cup B^l)  \text{ or } A^u \cup B^l = \emptyset \tag{b3}\\
\#(B^u) = \#(A^u \cup B^u)  \text{ or } A^u \cup B^u = \emptyset \tag{b4}
\end{align}

Since the sets are finite, the possibilities reduce to 

\begin{align}
A^l \subseteq B^l  \text{ or } A^l = B^l = \emptyset  \tag{b1+}\\
A^l \subseteq B^u  \text{ or } A^l = B^u = \emptyset \tag{b2+}\\
A^u \subseteq B^l  \text{ or } A^u = B^l = \emptyset \tag{b3+}\\
A^u \subseteq B^u  \text{ or } A^u = B^u = \emptyset \tag{b4+}
\end{align}

which again reduces to the three possibilities 
\begin{align}
A^u \cup B^u = \emptyset   \tag{1}\\
A^u \cup B^l = \emptyset \, \&\, B^u \neq \emptyset   \tag{2}\\
A^l \subseteq A^u \subseteq B^l \subseteq B^u   \tag{3}
\end{align}
 
\qed 
\end{proof}

\emph{These results provide for another level of abstraction from a numeric perspective that can actually help in deciding on solvability of the inverse problem in the \textsf{set HGOS} perspective}.

\begin{theorem}
If a collection of matrices $\mathcal{U}$ with entries in the interval $[0, 1]$ contains a matrix of the form \begin{equation}
\begin{pmatrix}
r_{ll} & r_{lu} \\
r_{ul} & r_{uu}
\end{pmatrix}
\end{equation}
not satisfying 
\begin{equation}
r_{ll} \leq r_{lu} \, \&\, r_{ul} \leq r_{uu} 
\end{equation}
or if it contains a matrix of the form 
\begin{equation}
\begin{pmatrix}
0 & 1 \\
1 & 1
\end{pmatrix}
\end{equation}
then the approximations cannot fit in a \textsf{set HGOS} scheme.
\end{theorem}

\subsection{Connections with Parameterized Approximation Spaces}

At face value, it is not possible to relate the approach of GRIF enhanced \textsf{GGS} variants with the approach used in parameterized approximation spaces because the latter is committed to pointwise approximations, and basic ideas of granularity used are different. Apart from philosophical correspondence of ideas, some   
interesting variations are shown to be possible in this subsection. These are relevant for comparison with the \emph{pilots algorithm} proposed in section \ref{plot}.

The setting of \cite{skajsd2016} is a \emph{parameterized approximation space} of the form \[\left\langle \underline{S}, \xi, h \right\rangle,\]  with 

$\xi: S \longmapsto \wp (S)$ being an uncertainty function ($S$ being interpreted as a set of objects) and $h : (\wp(S))^2 \longmapsto [0,1]$ a RIF. $\xi(x)$ is the set of things that are similar to $x$. Further, it is assumed that a set is definable if and only if it is a union of uncertain values. The authors however suggest the use of additional numeric functions to define $\xi$ - this aspect will be omitted in the proposed variants because it amounts to permitting more contamination.  $\xi(x)$ can also be read as a neighborhood of $x$. 

Pointwise approximations defined as below are used in the considerations 
\begin{align*}
(\forall X\in \wp (S) ) \, X^{low} = \{x: \, h (\xi (x), X) = 1\} \tag{low}\\
(\forall X\in \wp (S) ) \, X^{up} = \{x: \, 0 < h (\xi (x), X) \} \tag{up}
\end{align*}

Parametric granular approximations can be defined via (these are not used in \cite{js09}):
\begin{align*}
(\forall X\in \wp (S) ) \, X^{low} = \{\xi (x): \, h (\xi (x), X) = 1\} \tag{glow}\\
(\forall X\in \wp (S) ) \, X^{up} = \{\xi (x): \, 0 < h (\xi (x), X) \} \tag{gup}
\end{align*}

If $R \subseteq S^2$ is a tolerance, then approximations are defined in \cite{js1998} as  

\begin{align*}
(\forall X\in \wp (S) ) \, X^{low_R} = \{x: (\forall a)(Rxa \longrightarrow \, h (\xi (x), X) = 1)\}\\
(\forall X\in \wp (S) ) \, X^{up_R} = \{x: (\forall a)(Rxa \longrightarrow x: \, 0 < h (\xi (x), X)) \}
\end{align*}

The set $\{\xi(x): x\in S\}$ in the context of a parametric approximation spaces may be read as an anti  granulation from the perspective of the present author's axiomatic approach and also from the GDO (granules are definite objects) approach. 

If approximations are defined instead as 

\begin{align*}
(\forall X\in \wp (S) ) \, X^{low_Rg} = \bigcup \{\xi(x): (\forall a)(Rxa \longrightarrow \, h (\xi (x), X) = 1)\}\\
(\forall X\in \wp (S) ) \, X^{up_Rg} = \bigcup \{\xi (x): (\forall a)(Rxa \longrightarrow x: \, 0 < h (\xi (x), X)) \}
\end{align*}

then again the approximations are in general \emph{not granular}.  

In \cite{skajsd2016}, the authors expect two main conditions to be satisfied by a calculus of information granules (understood as per CGCP):

\begin{itemize}
\item {Granules or rather granulations should provide for compressed representation of complex nested clumps of objects like soft patterns, and}
\item {Concepts represented by granules should be robust with respect to their component deviations.  This is intended to mean that inclusion or closeness under a given construction is preserved by small deviations in components.}
\end{itemize}

For these conditions, it is necessary that more complex\\ information granules involved should be constructible from simpler ones together with relevant extensions of inclusion and closeness relations. For the latter condition satisfiability relations under rough mereology need to be invoked to define robustness in concrete terms. \emph{In the present author's view, the above mentioned idea of robustness of granules and granulations is suggestive of methods that ensure similarity of the semantic type of the granules. It can also be a weak principle of reducing contamination}.  

The necessity of identifying proper uncertainty functions that can be read as a \emph{mechanism of rough object identification} is central to the approach of \cite{js1998,skajsd2016}

\subsubsection{Granular Parametric Approximation Spaces}

Given a set CHGOS, and an additional mechanism of identification of uncertain objects it is possible to define additional approximations using GRIFs. This provides for granular generalization of the parametric approximation approach. 

Let \[\mathsf{X} \, =\, \left\langle \underline{\mathbb{O}},\, \underline{\mathbb{S}}, \gamma, \eta, l , u, \pc, \xi , \cup,  \cap, \bot, \top \right\rangle\] be a set CHGOS and $ \hslash: \mathbb{O} \longmapsto \wp (\mathbb{O})$ be a uncertainty map (defined by a formula possibly). For any corresponding element $X\in \mathbb{S}$ parametric lower and upper approximations can be defined (relative to a GRIF $\zeta_{\tau}$ and suitable bounds $1_o$ and $0_o$) as follows:
\begin{equation}
 X^{L_+} = \{\hslash(x):\, \zeta_{\tau}(\hslash(x), X) = 1_o  \}
\end{equation}
\begin{equation}
 X^{U_+} = \{\hslash(x):\, 0_o \prec \zeta_{\tau}(\hslash(x), X)\}
\end{equation}

It should be noted that the value of $1_o$ and $0_o$ depend on $\tau$ as has been demonstrated before.

\subsection{Extended Set-Theoretic Mereology}

\begin{definition}
By the \emph{natural partial order} on the set $M_Q$ will be meant the relation $\preceq$ defined by 
\begin{equation*}
(a_{ij}) \preceq (b_{ij}) \text{ if and only if } (\forall i, j)\, a_{ij} \leq b_{ij} 
\end{equation*}
\end{definition}

\begin{definition}
In a \textsf{set HGOS} $\mathbb{S}$, an element $A$ will be $r$-included in $B$ ($A \subseteq_r B$) if and only if $r \preceq \zeta^{\tau} (A, B)$ for a \textsf{GwqRIF} $\zeta$.   
\end{definition}

\begin{theorem}
In a \textsf{set HGOS} $\mathbb{S}$ with $\tau$ being a \textsf{RIF}, all of the following hold:
\begin{align*}
(\forall A, B, C)( A\subseteq_r B \, \&\,B\subseteq_q C \longrightarrow (\exists h )\, h\preceq r\, \&\, h\preceq q \, \&\, A\subseteq_h C) \\
(\forall A, B) (A\subseteq_h B \neq \bot \,\&\, \mathbf{0} \prec h\longrightarrow (\exists q) B\subseteq_q A \, \&\, \mathbf{0} \prec q)\\
(\forall A, B) (\pc AB \,\&\, C\subseteq_h A \longrightarrow  (\exists r) \, h\preceq r\, \& C\subseteq_r B)
\end{align*}
\end{theorem}
\begin{proof}
\begin{itemize}
\item {The first property is a consequence of the definition of $\zeta^{\tau}$ and the properties of the approximations assumed. It includes the case with $h=\bot$. So the property holds always and even when $\tau$ is a \textsf{wqRIF} and not a \textsf{RIF}.}
\item {$\mathbf{0} \prec h$ yields at least one of the entries in the matrix is non zero. This means some nonempty granules are part of the upper approximations of $A$ and $B$. This in turn yields the existence of a $q$ satisfying $\mathbf{0} \prec q$ and $B\subseteq_q A$. Condition R1 is assumed in this.}
\item {The third property holds if condition R0 is satisfied by $\tau$. This happens because it restricts the possible values of $\zeta^{\tau}(A, B)$.}
\end{itemize}
\qed 
\end{proof}

Clearly $\subseteq_r$ is more general than the parthood predicate of \textsf{RIF} based rough mereology and its fuzzy variants \cite{lp2011}. The associated logics that differ substantially from \cite{lpmsp2010} will appear separately for reasons of space.

\section{Extending GRIFs to GGS}\label{grifggs}

The natural question of extending \textsf{RIF}s and \textsf{GRIF}s to \textsf{GS} are explored in this section. Cardinality of elements of a \textsf{GS} are not defined by default, but measures similar to that can be obtained through correspondences that model analogy with \textsf{set HGOS}. Such correspondences may be regarded as a realization of second order rough set perspective because they involves assigning abstract collections of attributes to sets of attributes.

Let $\mathbb{S}$ and $\mathbb{W}$ be two \textsf{GGS} of the following forms:
\begin{align}
\mathbb{S} \, =\, \left\langle \underline{\mathbb{S}}, \mathcal{G}, l , u, \pc, \leq , \vee,  \wedge, \bot, \top \right\rangle    \tag{1}\\
\mathbb{W} \, =\, \left\langle \underline{\mathbb{W}}, \mathcal{F}, l , u, \pc, \leq ,  \vee,  \wedge, \bot, \top \right\rangle   \tag{2}
\end{align}
 then a map $\varphi : \mathbb{S} \longmapsto \mathbb{W}$ is a \emph{morphism} if it satisfies (it is assumed that the reader can figure out the intended interpretation of operations):
 
 \begin{align}
(\forall a \in \mathbb{S})\,\varphi(a^l) = (\varphi(a))^l \,\&\, \varphi(a^u) = (\varphi(a))^u    \tag{lu-morphism}\\
(\forall a, b \in \mathbb{S}) \, (\pc ab \longrightarrow \pc \varphi(a)\varphi (b))   \tag{$\pc$-morphism}\\
(\forall a, b \in \mathbb{S})\, ( a \leq b \longrightarrow  \varphi(a) \leq \varphi (b))   \tag{$\leq$-morphism}\\
\varphi (a \vee b)\stackrel{\omega}{=}   \varphi (a) \vee \varphi (b) \tag{weak $\vee$-morphism}\\
\varphi (a \wedge b)\stackrel{\omega}{=}   \varphi (a) \wedge \varphi (b)  \tag{weak $\wedge$-morphism}\\
\varphi(\bot) = \bot \, \& \varphi (\top) = \top   \tag{0}
\end{align}

\begin{proposition}
In the above context, if $\mathbb{W}$ is a \textsf{set HGOS} (expressed with a superfluous signature), then the condition $\pc$-morphism coincides with $\leq$-morphism. 
\end{proposition}

Note that $\varphi$ need not be a closed morphism. This is intended to permit relatively loose interpretations of \textsf{GGS} in \textsf{set HGOS} for the purpose of defining weaker concepts of cardinalities. 

\emph{The following measures and generalized cardinalities are motivated by relative values that can potentially evaluate the strength of sets of attributes. This can again be relative to approximations.}

\begin{definition}\label{card}
If there exists a morphism $\varphi$ from a \textsf{GGS} $\mathbb{S}$ to a \textsf{set HGOS} $\mathbb{W}$, then the $\varphi$-\emph{cardinality} ($\#_{\varphi}$)of an element $a\in\mathbb{S} $ will be taken to be $\# (\varphi (a))$, and in addition the \textsf{GGS} will be said to be \emph{have numeracy} (GGSN for short). If $\varphi$ is a closed morphism, then $\#_{\varphi}(a)$ will be said to be \emph{closed}, and in addition the \textsf{GGS} will be said to \emph{have closed numeracy} (GGSCN for short). By analogy, a \textsf{GGSN} that is also a \textsf{GS} will be termed a \textsf{GSN}.  
\end{definition}

\begin{problem}
In the context of definition \ref{card}, the problem of existence of closed morphisms can be involved. What simple conditions ensure the existence of such morphisms? 
\end{problem}

All definitions of section \ref{granrif} can be extended to a \textsf{GGSN} along the following lines:

\begin{definition}\label{grifggsn}
If $\mathbb{S}$ is a \textsf{GGSN}, $A, B\in \mathbb{S}$, $\sigma, \pi \in \{l, u\}$, $\varphi : \mathbb{S} \longmapsto $  and the denominator in the expression is non zero, then let 
\begin{equation*}
\nu_{\sigma \pi}(A, B ) = \dfrac{\#(\varphi(A^\sigma) \cap \varphi(B^\pi) )}{\#(\varphi(A^\sigma) )}  \tag{$\sigma\pi$-grif1} 
\end{equation*}
If $\# (\varphi(A^\sigma)) = 0 $, then set the value of $\nu_{\sigma \pi}(A, B)$ to $1$. 
The function $\nu_{\sigma \pi}$ will be referred to as a $\varphi$-\emph{hasty GRIF}.
\end{definition}

The adjective \emph{hasty} is relative to the following definition in a \textsf{GSN} 

\begin{definition}\label{grifgsn}
If $\mathbb{S}$ is a \textsf{GSN}, $A, B\in \mathbb{S}$, $\sigma, \pi \in \{l, u\}$, $\varphi : \mathbb{S} \longmapsto $  and the denominator in the expression is non zero, then let 
\begin{equation*}
\nu_{\sigma \pi}(A, B ) = \dfrac{\#(\varphi(A^\sigma  \wedge B^\pi))}{\#(\varphi(A^\sigma) )}  \tag{$\sigma\pi$-grif2} 
\end{equation*}
If $\# (\varphi(A^\pi)) = 0 $, then set the value of $\nu_{\sigma \pi}(A, B)$ to $1$. 
The function $\nu_{\sigma \pi}$ will be referred to as a $\varphi$-\emph{GRIF}.
\end{definition}

Clearly $\varphi$-\textsf{GRIF}s may not be definable in a \textsf{GGSN}. 

\begin{proposition}
In a \textsf{GGSN} $\mathbb{S}$, if $\nu_{\sigma \pi}$ is a hasty $\varphi$-\textsf{GRIF} and $\varphi$ is a closed morphism, then $\nu_{\sigma \pi}$ is also a $\varphi$- GRIF.
\end{proposition}

\begin{proof}
Let $\varphi: \mathbb{S} \longmapsto \mathbb{W}$ be the closed morphism from the \textsf{GGSN} $\mathbb{S}$ into the \textsf{set HGOS} $\mathbb{W}$.  
 
For any two elements $A, B \in \mathbb{S}$,if $\varphi(A) \cap \varphi(B)$ is defined, then $\varphi(A \wedge B)$ must also be defined and  for any $\sigma, \pi \in \{l, u\}$, if $\#(\varphi (A)) \neq 0$, 
\begin{equation*}
\nu_{\sigma \pi}(A, B ) = \dfrac{\#(\varphi(A^\sigma) \cap \varphi(B^\pi) )}{\#(\varphi(A^\sigma) )} =  \dfrac{\#(\varphi(A^\sigma \wedge B^\pi) )}{\#(\varphi(A^\sigma) )}.
\end{equation*} 
 
So the proposition holds. 
\qed
\end{proof}

\subsection{Inverse Problem: Partial Solutions}

In this section a method for checking whether a proposed solution to an inverse problem is admissible or not is proposed through granular RIFs and variants. This is important because the possible solution space is bound to appear too wide in a number of cases, and even when the target is restricted at a theoretical level (through axioms/conditions) too many models may satisfy or become difficult to construct in the first place. It is worthwhile to distinguish between the two situations because they require distinct solution strategies.   

\textbf{Case-1}: The cardinalities of the objects approximated and the universe are not known.
In this case, the following steps are appropriate.

\begin{itemize}
\item {Using some strategy, compute GRIfs. }
\item {Estimate an universe on the basis of the computed GRIfs and possible order.}
\item {Construct possible models on the basis of this and eliminate those that are incompatible with the order suggested by the GRIF.}
\end{itemize}

Thus from a dataset consisting of general lower and upper approximations and some order relations, it may be possible to arrive at a  smaller collection of possible models that represents the actual situation through GRIFs.

\textbf{Case-2}: The universe is known, but the cardinalities of the objects approximated are not fully known.

In this case, the following steps are appropriate.
\begin{itemize}
\item {Using some strategy, compute GRIfs. }
\item {Construct possible models on the basis of this and eliminate those that are incompatible with the order suggested by the GRIF.}
\end{itemize}

Again from a dataset consisting of general lower and upper approximations and some order relations, it may be possible to arrive at a  smaller collection of models that represents the actual situation through GRIFs.

Part of the complexity of proposed strategies is partly illustrated in the following example (based on the example in \cite{am9114}) for relation based rough sets.

\begin{example}
 
Let $ S\,=\, \{a, b, c, e, f\}$ and let $R$  be a binary relation on it 
defined via
\begin{align*}
R\,=\,& \{(a,\,a),\,(b,\, b),\, (c,\, c),\, (a,\,b), \\ 
 & (c,\,e),\,(e,\,f),\,(e,\,c),\,(f,\,e),\,(e,\,b) \}
\end{align*}

$R$ is not reflexive, transitive, symmetric or anti-symmetric. 

The table for granules (successor neighborhoods) is 
\begin{table}[hbt]
\centering
\caption{Successor Neighborhoods}\label{nbd69}
\begin{tabular}[hbt]{c|c|c|c|c|c}
\toprule
\textbf{Objects} $\mathbf{E}$ & $a\,$ & $ b\,$ & $c\,$ & $e\,$ & $f \,$\\
\toprule
\textbf{Neighborhoods} $\mathbf{[E]}$ & $\{a \}$ & $\{a, b, e \}$ & $\{c, e \}$ & $\{ c,f \}$ &$\{e \}$\\
\bottomrule
\end{tabular}
\end{table}

The approximations and rough objects are computed in Table \ref{exp69} (strings of letters of the form $abe$ are intended as abbreviation for the subset $\{a, b, e \}$ and $\lrcorner$ is for $,$ among subsets)

\begin{table}[hbt]
\centering
\caption{Approximations and Rough Objects}\label{exp69}
\begin{tabular}[hbt]{c|c|c|c}
\toprule
\textbf{Rough Object} $\,\mathbf{x}\,$  & $\,\mathbf{z^l}\, $ & $\, \mathbf{z^u}\,$ & \textbf{RO Identifier} \\
\toprule
$\{a \lrcorner b\lrcorner ab\}$ & $\{a \} $ & $\{abe \}$ & $\{ 3\} $ \\
\midrule
$\{ae\lrcorner abe \}$ & $\{a \} $ & $\{abce \}$ & $\{6 \} $ \\
\midrule
$\{e\lrcorner be \}$ & $\{e \} $ & $\{ abec\}$ & $\{9 \} $ \\
\midrule
$\{c \}$ & $\{\emptyset \} $ & $\{cef \}$ & $\{ 15 \} $ \\
\midrule
$\{f \}$ & $\{\emptyset \} $ & $\{ cf \}$ & $\{24 \} $ \\
\midrule
$\{cf \}$ & $\{cf \} $ & $\{cef \}$ & $\{27 \} $ \\
\midrule
$\{ bc\lrcorner bf\}$ & $\{\emptyset \} $ & $\{ S\}$ & $\{30 \} $ \\
\midrule
$\{ac\lrcorner af \lrcorner abc \lrcorner abf \}$ & $\{a \} $ & $\{S \}$ & $\{ 33\} $ \\
\midrule
$\{aef \}$ & $\{ae \} $ & $\{S \}$ & $\{36 \} $ \\
\midrule
$\{ef\lrcorner bef \}$ & $\{e \} $ & $\{S \}$ & $\{ 42\} $ \\
\midrule
$\{ec\lrcorner bce \}$ & $\{ec \} $ & $\{S \}$ & $\{ 45\} $ \\
\midrule
$\{bcf \}$ & $\{ fc\} $ & $\{S \}$ & $\{51 \} $ \\
\midrule
$\{ abef\}$ & $\{abe \} $ & $\{S \}$ & $\{54 \} $ \\
\midrule
$\{ ace\}$ & $\{ace \} $ & $\{S \}$ & $\{ 60\} $ \\
\midrule
$\{acf \}$ & $\{acf \} $ & $\{S \}$ & $\{63 \} $ \\
\midrule
$\{ecf\lrcorner bcef \}$ & $\{ cef\} $ & $\{S \}$ & $\{69 \} $ \\
\midrule
$\{abcf \}$ & $\{abcf \} $ & $\{S \}$ & $\{72 \} $ \\
\midrule
$\{abce \}$ & $\{ abcf\} $ & $\{S \}$ & $\{78 \} $ \\
\bottomrule
\end{tabular}
\end{table}

Under the rough inclusion order, the bounded lattice of rough objects can be found in \cite{am9411}. Now suppose that inform relating to the rough objects (and subsets included) $3, 30, 45$ and sets $ef\lrcorner   ae\lrcorner    acf\lrcorner   abef \lrcorner   ac$, and related GRIFs generated over the standard rough inclusion is available. The first problem is of finding the extent to which Table \ref{exp69} can be reconstructed. As many as $132$ GRIF matrices would be generated by the sets under consideration \[ec\lrcorner bce\lrcorner bc\lrcorner bf \lrcorner a \lrcorner b \lrcorner ab \lrcorner ef \lrcorner ae \lrcorner acf \lrcorner abef \lrcorner ac  \]

For example, with $\tau$ being the usual rough inclusion function,
\begin{equation}
\zeta^{\tau} (ab, acf) = \begin{pmatrix}
1 & 1 \\
\frac{1}{3} & 1
\end{pmatrix}
\end{equation}

In the considerations $\zeta^{\tau}(ab, be)$ is not known among things. Estimating such unknowns can be done with suitable aggregation and commonality operations described in the previous sections. Obviously this requires additional justification.

\end{example}

The above example clearly motivates the following subproblem that is also potentially relevant in cryptography or in hiding information.

\begin{problem}
Given a subset of rough objects, and approximations in a known universe, and associated GRIFs, when and how can a set HGOS that fits the information be constructed? 
\end{problem}

\begin{remark}
In student-centered learning students are put at the center of the learning process, and are encouraged to learn through active methods. Arguably, students become more responsible for their learning in such environments. In traditional teacher-centered classrooms, teachers have the role of instructors and are intended to function as the only source of knowledge. By contrast, teachers are typically intended to perform the role of facilitators in student-centered learning contexts. A number of best practices for teaching in such contexts \cite{jrp2016} have evolved over time. It may be noted that the impact of AI on enhancing classroom learning and learning in general has been very limited (see \cite{chos2018} and related references).

The inverse problem in practical contexts as in student-centered learning (explored by the present author in a forthcoming paper) typically has additional information available. This motivates hybrid strategies in the situation and justifies the use of t-norms and s-norms to an extent.
\end{remark}

\section{Pilots Algorithm and Other Applications}\label{plot}

This new algorithm is closely related to how pilots fly modern airplanes under several constraints, therefore it has been named as \emph{pilots have limited will algorithm} or \emph{pilots algorithm}. It involves use of GRIFs. Even when GRIFs are replaced by RIFs, the algorithm schema is not known. The purpose of the algorithm is to prevent pilots from taking bad decisions in abnormal situations. In special cases, it can even be used for handling concepts of closeness, and betweenness to a degree in robot navigation \cite{olp2012}.   

Modern airplanes are largely driven by safety-critical software, and pilots essentially manage a limited abstraction. The software is almost always tied to the hardware (specific model of the plane), and has high complexity. Even when autopilots are switched off, pilots are permitted limited control, and erroneous manual inputs may be altered by the underlying software - in other words, this is about active control technology (ACT) in modern aviation.  Further pilots, first officers, and flight engineers are provided with detailed guide books for handling different adverse situations. Pilots learn on simulators, on flights and by studying failures - this may be insufficient for handling planes with ACT (see \cite{mbt96}).

In the present author's opinion it is very important to model the whole system including the pilots, first officers and flight engineers for achieving perfect safety criticality -- this can also be related to \emph{job safety analysis} in the context of safety management systems used in aviation (see \cite{sms2008}). Current tools used in such systems do not focus on rigorous models of vague aspects of training and other behavior of human participants. Further, competition between flight operators has led to a number of hazardous operational strategies that bypass standard operational procedures. It is also a fact that a majority of plane crashes have happened due to technical faults/failures in hardware and software, and standard violations. 

Failure of hardware components can cause sensors to report intractable wrong values or not reporting anything at all. In these conditions, even standard manoeuvres may cause the airplane to crash.

The proposed algorithm is restricted to snapshots of in-flight actions to enable a relatively static understanding of the process. Essentially the pilots notice errors, fail to rectify them at first attempt and manage to do so in a second. This scenario means that something is wrong with the plane. A goal of the algorithm is to provide optimal suggestions for improving the quality of decisions taken during an event. Functions $\tau$ and $\zeta_{\tau}$ are assumed to be definable. The evaluated GRIFs used are expected to be computed by a dedicated software that has access to more information about technical problems faced by the plane than the pilots have access to. The schematics of the context relative to the software is illustrated in Figure \ref{schematics} -- the lines may be read from top to bottom as \emph{is aware of}.

\begin{center}
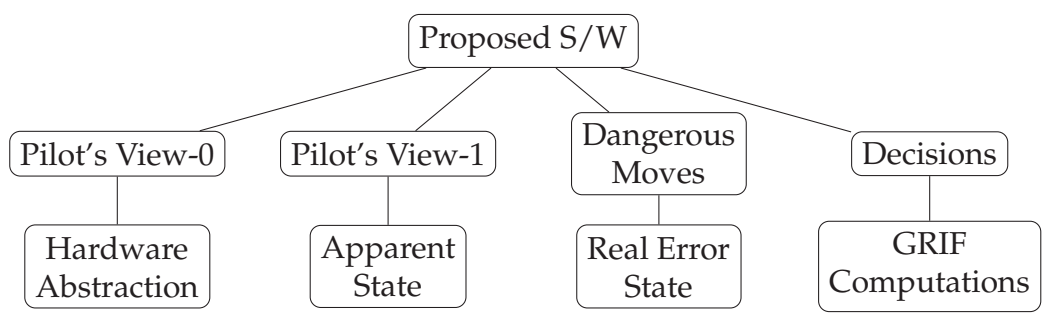
\begin{figure}[hbt]
 \begin{tikzpicture}[sibling distance=3.6cm,
  every node/.style = {shape=rectangle, rounded corners,
    draw, align=center,}]
  \node {Proposed S/W}
    child { node {Pilot's View-0} 
      child { node {Hardware\\ Abstraction}}}
    child { node {Pilot's View-1} 
      child { node {Apparent\\ State}}}
    child { node {Dangerous\\Moves}
      child { node {Real Error\\ State}}}
      child { node {Decisions}
        child { node {GRIF\\ Computations} }};
\end{tikzpicture}
\caption{Schematics of the Pilot's Algorithm Context}\label{schematics}
\end{figure}
\end{center}

\begin{enumerate}
\item{\textsf{State: Plane is in flight at some altitude $X$ m with speed $Z$ knots}} 
\item{\textsf{Pilots notice error indication $Err$ and alarm $Al$}}
\item{\textsf{Pilots constructs approximations $Er^l$ and $Er^u$ approximating relevant state $Erp$ (defined by the plane) } - these approximations are abstract}
\item{\textsf{Pilots understanding of closeness to real situation is the evaluated GRIF $\zeta_{\tau}(Er, Erp)$}}
\item{\textsf{Pilots propose approximations $Sok^l$ and $Sok^u$ that correspond to error resolved state $Sokp$ (defined by the plane)}  - these approximations are abstract}
\item{\textsf{Pilots understanding of closeness to real situation is the evaluated GRIF $\zeta_{\tau}(Sok,Sokp)$}}
\item{\textsf{Possible actions dictated by granularity are $A_1, \,\ldots,\,  A_n$}}
\item{\textsf{Pilots perform action $A_1$}}
\item{\textsf{Pilots notice error indication $Err$ and alarm $Al$}}
\item{\textsf{Pilots propose revised approximations $Sok_1^l$ and $Sok_1^u$ that correspond to revised error resolved state $Sokp_1$}  - these approximations are abstract}
\item{\textsf{Pilots understanding of closeness to real situation is the evaluated GRIF $\zeta_{\tau}(Sok_1,Sokp_1)$ that is better than $\zeta_{\tau}(Sok,Sokp)$. }}
\item{\textsf{Possible actions dictated by granularity are $C_1, \,\ldots,\,  C_n$}}
\item{\textsf{Pilots perform action $C_1$}}
\item{\textsf{Flight becomes stable}}
\end{enumerate}

In the above, actions $C_i$ and $A_i$ can also be quasi ordered for decision making in steps $9$ and $14$. The quality of approximations constructed by pilots depends on the error condition, and pilot's experience (in the third step). 

Decision making would suffer if RIFs are used instead in the algorithm. While this claim has been justified theoretically, empirical comparison is hindered by nonavailability of crash related datasets at the level desired. Internal hardware designs and software are specific to the plane model and are proprietary. For this reason an abstract minimalist dataset is suggested in the following subsection for comparison.

\subsection{Data Set Construction}

The dataset should consist of objects $\mathbb{S}$ with associated labels, subsets of a powerset of attributes $\mathbb{S}$, and related valuations, subject to the further condition that $\mathbb{S}$ is a union of two subsets $\mathbb{A}$ and $\mathbb{B}$. Assume that $\mathcal{G} \subseteq \mathbb{B}$. Elements of $\mathbb{A}$ are required to have their approximations in $\mathbb{B}$. Further let

\begin{itemize}
\item {$C := \{c_i : i\in 1, \ldots , n  \}$ is intended as the set of elements that can be written in the form $x^l$ or $x^u$.}
\item {$A:= \{a_i : i \in 1, \ldots , r\}$}
\item {$\{b_i : i \in 1 , \ldots , q\}\cup C = B$}
\item {the granulation $\mathcal{G} := \{g_i : i \in 1, \ldots l \} \subseteq \mathbb{S}$}
\item {$A\cap B = \emptyset$}
\item {$\#(C) = n$, $\#(A) = r$, $\#(B) = q+n$ and $\#(\mathcal{G}) = l =< n$}
\end{itemize}

Every element of $C$ is assumed to be related to a data table, while approximations of elements of $A$ are alone known. It is easy to satisfy the requirement that $\mathbb{S}$ along with objects, and additional operations form a CGGS.

\subsection{Remarks on Other Application Contexts}

Two other distinct application contexts that show that granular \textsf{RIF}s can be way better than \textsf{RIF}s perceived in the Skowron-Polkowski style mereological approach. 

For broad overviews of computational linguistics, the reader is referred to \cite{sepsl,cfl2010}. Natural languages expressed in a written script can be viewed as a set of strings in alphabets that are classifiable into words, sentences, clauses, phrases and other linguistic categories with the help of a rule set based on occurrence of particular distinguished symbols like white space and punctuation marks.  
In probabilist approaches these are represented through linear ngrams subject to units being characters, words, or through syntactic ngrams. For example, a linear 5-gram of words would be a sequence of five words. Syntactic ngrams, include those based on dependency relations among parts and part of speech ngrams (that are defined as subsequences of contiguous overlapping part-of-speech sequences with text size $n$).

Many problems of computational linguistics involve some method of identifying similar expressions in the form of\\ ngrams in the context in question. For example, the problem of \emph{referring expression generation} (REG)\cite{ekkd2012} concerns the production of a description of an entity (from a dynamic dataset) that enables the hearer to identify that entity in a given context. The first step towards solving such problems concern the selection of a suitable form of referring expression. Often it is about descriptions like \emph{the white four wheeler} or \emph{a big cat} or \emph{the walking bovine divinity}.

RIFs can be used to reduce statements of the form \emph{$A$ is similar to $B$} to \emph{$A$ is roughly included to the degree $r$ in $B$}, or to the form \emph{the inclusion degree of $A$ in $B$ is $r$}. A far better idea would be to use \textsf{GRIF}s in the scenario because it is easier to reason with relatively definite approximations of $A$ and $B$. 

Functional approximations of the predicate \textbf{includes the meaning of} and those having the form \textbf{all meanings of A are included in B} are also of interest. These can be handled through choice inclusive lower and upper approximations \cite{am99}. Given a set  of potential synonyms of a word, it is easy to see that even the condition $A^{ll} = A^l$ can fail. 

\subsubsection{Weights and Orders}

The act of regarding all attributes as having equal value in the constructive description of objects is known to be problematic. In dominance based rough sets \cite{gms2007}, this is addressed partially through orders and order based ranking of attributes. The approach is also related to theory of pairwise comparison. When specific covers of the attribute set are of interest, multiple weights may be assigned to attributes based on the element of the cover they belong to. For example, if attribute $x$ belongs to $A$ and $B$ that are in the cover $\mathcal{S}$, then it can be reasonable to assign distinct numeric weights $w(x, A)$ and $w(x, B)$ to the attribute $x$. In these scenarios, the weights corresponding to each element of the cover form a chain. \textsf{GwqRIF}s can be easily extended to handle such weighted chain decomposition of the attribute set.

The present author is also involved in the analysis of a dataset on health care access of women in Kolkata urban conglomerate. A few health care indices (at different stages of integration) that rely on two levels of weighting and chain decomposition of the attribute set have been proposed by her in a forthcoming paper.

\subsubsection{Connection with Generalized Probability Theories}

While the function mentioned in the introduction can be read as something analogous to conditional probability and Bayesian methods using \textsf{RIF}s justified (see for example \cite{sl2006}, it is known that many theoretical concepts cannot be translated between general rough set and subjective probability theories \cite{am9411,am9501}. 
In all approaches to probabilities, the concept of an exact object is subjective. If they are fixed in advance, then concepts of upper approximation are actualizable. If the generalized probability takes real values (or in a suitable ordered structure), then it is possible to define upper, lower and other approximations through constraints based on cut-off values \cite{yy10}. As no proper granularity is admissible in a probability theory, a reasonable analogy remains elusive.  

\section{Further Directions and Open Problems}

In this research, the concept of high granular operator spaces and variants are improved substantially. As a result of this, these can be studied from a purely partial algebraic system perspective. These are illustrated with many examples, and connections with information systems (tables) are also explored. The inverse problem of general rough sets is provided with a clean partial algebraic system formalism. 

To possibly relate \textsf{GGS} and variants with rough mereology based on inclusion degrees, reduce contamination in practical situations, generalized variants of RIFs, and  granular versions of rough inclusion functions are proposed and characterized in some detail by the present author.  Representation theorems on the form of matrices are also proved by her. Aggregation and commonality operations are shown to be possible in the settings. The ideal representation proposed motivates a number of existential and mathematical problems. In particular, the best t-norm and s-norm that attain the ideal representation can be interpreted as a fuzzy perspective that ensures the representation. These and connection of the approach with subjective probability and belief theory will be part of a forthcoming paper by the present author. Further studies on related logics are also motivated by this research. 

The granular rough inclusion functions are used in a new algorithm for approximate decision making in human - machine interaction contexts (safety critical). It is assumed that the humans involved have substantial knowledge gaps about internal workings (especially about failure states) of the machine. It will be useful to have publicly available real datasets for the problem class. GRIFs are also shown to be applicable in solving inverse problems.

It is not clear as to how theoretical approaches may be connected with practical algorithms (an issue mentioned in \cite{yy2015}). Aspects of this are addressed in this research.

In a forthcoming paper, it is shown by the present author that even Pre \textsf{GGS} are equivalent to certain single sorted partial algebras. The algebras are obviously general enough to cover generalized versions of rough and fuzzy sets, but the exact scope of applications in computational intelligence remains to be explored. From the perspective of formal logics, this is a ground breaking result.

\begin{flushleft}
\begin{small}
\textbf{Acknowledgement}: 
\end{small}
\end{flushleft}
The present author would like to thank Ivo D{\"u}ntsch for reading this paper, and Anna Gomolinska, and others for comments on an earlier draft of this paper.

\bibliographystyle{model1b-num-names}
\bibliography{algroughf6}

\begin{thebibliography}{77}
\expandafter\ifx\csname natexlab\endcsname\relax\def\natexlab#1{#1}\fi
\providecommand{\bibinfo}[2]{#2}
\ifx\xfnm\relax \def\xfnm[#1]{\unskip,\space#1}\fi
%Type = Inproceedings
\bibitem[{Agrawal et~al.(1996)Agrawal, Mannila, R., Toivonen and
  Verkamo}]{amst96}
\bibinfo{author}{R.~Agrawal}, \bibinfo{author}{H.~Mannila},
  \bibinfo{author}{S.~R.}, \bibinfo{author}{H.~Toivonen},
  \bibinfo{author}{A.~Verkamo}, \bibinfo{title}{{Fast Discovery of Association
  Rules}}, in: \bibinfo{editor}{others} (Ed.), \bibinfo{booktitle}{{Advances in
  Knowledge Discovery and Data Mining}}, \bibinfo{publisher}{MIT Press},
  \bibinfo{year}{1996}, pp. \bibinfo{pages}{307--328}.
%Type = Book
\bibitem[{Alsina et~al.(2006)Alsina, Frank and Schweizer}]{cjb2006}
\bibinfo{author}{C.~Alsina}, \bibinfo{author}{J.M. Frank},
  \bibinfo{author}{B.~Schweizer}, \bibinfo{title}{{Associative Functions:
  Triangular Norms and Copulas}}, \bibinfo{publisher}{World Scientific},
  \bibinfo{year}{2006}.
%Type = Article
\bibitem[{Banerjee and Chakraborty(1996)}]{bc1}
\bibinfo{author}{M.~Banerjee}, \bibinfo{author}{M.K. Chakraborty},
  \bibinfo{title}{{Rough Sets Through Algebraic Logic}},
  \bibinfo{journal}{Fundamenta Informaticae} \bibinfo{volume}{28}
  (\bibinfo{year}{1996}) \bibinfo{pages}{211--221}.
%Type = Book
\bibitem[{Burmeister(2002)}]{bu}
\bibinfo{author}{P.~Burmeister}, \bibinfo{title}{{A Model-Theoretic Oriented
  Approach to Partial Algebras}}, \bibinfo{publisher}{Akademie-Verlag},
  \bibinfo{year}{1986, 2002}.
%Type = Incollection
\bibitem[{Cattaneo(2018)}]{gc2018}
\bibinfo{author}{G.~Cattaneo}, \bibinfo{title}{{Algebraic Methods for Rough
  Approximation Spaces by Lattice Interior--closure Operations}}, in:
  \bibinfo{editor}{A.~Mani}, \bibinfo{editor}{I.~D{\"u}ntsch},
  \bibinfo{editor}{G.~Cattaneo} (Eds.), \bibinfo{booktitle}{{Algebraic Methods
  in General Rough Sets}}, {Trends in Mathematics},
  \bibinfo{publisher}{Springer International}, \bibinfo{year}{2018}, pp.
  \bibinfo{pages}{13--156}.
%Type = Article
\bibitem[{Chakraborty(2014)}]{mkc2014}
\bibinfo{author}{M.K. Chakraborty}, \bibinfo{title}{{Membership Function Based
  Rough Set}}, \bibinfo{journal}{Information Sciences} \bibinfo{volume}{55}
  (\bibinfo{year}{2014}) \bibinfo{pages}{402--411}.
%Type = Article
\bibitem[{Chen et~al.(2014)Chen, Li, Zhang and Kwong}]{chen2014}
\bibinfo{author}{D.~Chen}, \bibinfo{author}{W.~Li}, \bibinfo{author}{X.~Zhang},
  \bibinfo{author}{S.~Kwong}, \bibinfo{title}{{Evidence Theory Based Numerical
  Algorithms Of Attribute Reduction With Neighborhood Covering Rough Sets}},
  \bibinfo{journal}{Int. J. Approx. Reasoning} \bibinfo{volume}{55}
  (\bibinfo{year}{2014}) \bibinfo{pages}{908--923}.
%Type = Incollection
\bibitem[{Chorney(2018)}]{chos2018}
\bibinfo{author}{S.~Chorney}, \bibinfo{title}{{Digital Technology in Teaching
  Mathematical Competency: A Paradigm Shift}}, in:
  \bibinfo{editor}{A.~Kajander}, et~al. (Eds.), \bibinfo{booktitle}{{Advances
  in Mathematics Education}}, \bibinfo{publisher}{Springer International},
  \bibinfo{year}{2018}, pp. \bibinfo{pages}{245--256}.
%Type = Article
\bibitem[{Ciucci(2009)}]{cd3}
\bibinfo{author}{D.~Ciucci}, \bibinfo{title}{{Approximation Algebra and
  Framework}}, \bibinfo{journal}{Fundamenta Informaticae} \bibinfo{volume}{94}
  (\bibinfo{year}{2009}) \bibinfo{pages}{147--161}.
%Type = Incollection
\bibitem[{Ciucci(2017)}]{cd2017}
\bibinfo{author}{D.~Ciucci}, \bibinfo{title}{{Back To The Beginnings: Pawlak'S
  Definitions of The Terms Information System and Rough Set}}, in:
  \bibinfo{editor}{G.~Wang}, et~al. (Eds.), \bibinfo{booktitle}{{Thriving Rough
  Sets}}, {Studies in Computational Intelligence 708},
  \bibinfo{publisher}{Springer International}, \bibinfo{year}{2017}, pp.
  \bibinfo{pages}{225--236}.
%Type = Book
\bibitem[{Clark et~al.(2010)Clark, Fox and Lappin}]{cfl2010}
\bibinfo{editor}{A.~Clark}, \bibinfo{editor}{C.~Fox},
  \bibinfo{editor}{S.~Lappin} (Eds.), \bibinfo{title}{{The Handbook of
  Computational Linguistics and Natural Language Processing}},
  \bibinfo{publisher}{Wiley Blackwell}, \bibinfo{year}{2010}.
%Type = Article
\bibitem[{Duda and Chajda(1977)}]{jc1977}
\bibinfo{author}{J.~Duda}, \bibinfo{author}{I.~Chajda}, \bibinfo{title}{{Ideals
  of Binary Relational Systems}}, \bibinfo{journal}{Casopis pro pestovani
  matematiki} \bibinfo{volume}{102} (\bibinfo{year}{1977})
  \bibinfo{pages}{280--291}.
%Type = Article
\bibitem[{D{\"u}ntsch and Gediga(1997)}]{idgg1997}
\bibinfo{author}{I.~D{\"u}ntsch}, \bibinfo{author}{G.~Gediga},
  \bibinfo{title}{{Algebraic Aspects of Attribute Dependencies in Information
  Systems}}, \bibinfo{journal}{Fundamenta Informaticae} \bibinfo{volume}{29}
  (\bibinfo{year}{1997}) \bibinfo{pages}{119--133}.
%Type = Book
\bibitem[{D{\"u}ntsch and Gediga(2000)}]{gdu}
\bibinfo{author}{I.~D{\"u}ntsch}, \bibinfo{author}{G.~Gediga},
  \bibinfo{title}{{Rough set data analysis: A road to non-invasive knowledge
  discovery}}, \bibinfo{publisher}{Methodos Publishers}, \bibinfo{year}{2000}.
%Type = Article
\bibitem[{D{\"u}ntsch and Or{\l}owska(2004)}]{ideo2004}
\bibinfo{author}{I.~D{\"u}ntsch}, \bibinfo{author}{E.~Or{\l}owska},
  \bibinfo{title}{{Boolean algebras Arising from Information Systems}},
  \bibinfo{journal}{Annals of Pure and Appl. Logic} \bibinfo{volume}{127}
  (\bibinfo{year}{2004}) \bibinfo{pages}{77--98}.
%Type = Incollection
\bibitem[{Gomolinska(2008)}]{ag3}
\bibinfo{author}{A.~Gomolinska}, \bibinfo{title}{{On Certain Rough Inclusion
  Functions}}, in: \bibinfo{editor}{J.F. Peters}, et~al. (Eds.),
  \bibinfo{booktitle}{{Transactions on Rough Sets IX, LNCS 5390}},
  \bibinfo{publisher}{Springer Verlag}, \bibinfo{year}{2008}, pp.
  \bibinfo{pages}{35--55}.
%Type = Article
\bibitem[{Gomolinska(2009)}]{ag2009}
\bibinfo{author}{A.~Gomolinska}, \bibinfo{title}{{Rough Approximation Based on
  Weak q-RIFs}}, \bibinfo{journal}{Transactions on Rough Sets}
  \bibinfo{volume}{X} (\bibinfo{year}{2009}) \bibinfo{pages}{117--135}.
%Type = Article
\bibitem[{Greco et~al.(2007)Greco, Matarazzo and Slowin{\'s}ki}]{gms2007}
\bibinfo{author}{S.~Greco}, \bibinfo{author}{B.~Matarazzo},
  \bibinfo{author}{R.~Slowin{\'s}ki}, \bibinfo{title}{{Dominance Based Rough
  Set Approach}}, \bibinfo{journal}{Transactions on Rough Sets}
  \bibinfo{volume}{VII} (\bibinfo{year}{2007}) \bibinfo{pages}{36--52}.
%Type = Article
\bibitem[{Gruszczy{\'n}ski(2013)}]{rafal2013}
\bibinfo{author}{R.~Gruszczy{\'n}ski}, \bibinfo{title}{{Mereological Fusion As
  An Upper Bound}}, \bibinfo{journal}{Bulletin of The Section of Logic}
  \bibinfo{volume}{42} (\bibinfo{year}{2013}) \bibinfo{pages}{135--149}.
%Type = Article
\bibitem[{Gruszczy{\'n}ski and Varzi(2015)}]{rgac15}
\bibinfo{author}{R.~Gruszczy{\'n}ski}, \bibinfo{author}{A.~Varzi},
  \bibinfo{title}{{Mereology Then and Now}}, \bibinfo{journal}{Logic and
  Logical Philosophy} \bibinfo{volume}{24} (\bibinfo{year}{2015})
  \bibinfo{pages}{409--427}.
%Type = Incollection
\bibitem[{Hoh and Mitchell(2018)}]{mbt96}
\bibinfo{author}{R.~Hoh}, \bibinfo{author}{D.G. Mitchell},
  \bibinfo{title}{{Handling-Qualities Specification - a Functional Requirement
  for the Flight Control System}}, in: \bibinfo{editor}{M.B. Tischler} (Ed.),
  \bibinfo{booktitle}{{Advances in Aviation Flight Control}},
  \bibinfo{publisher}{Taylor and Francis, London}, \bibinfo{year}{1996,2018},
  pp. \bibinfo{pages}{3--34}.
%Type = Book
\bibitem[{Husson et~al.(2011)Husson, L{\^e} and Pag{\`e}s}]{fsj2011}
\bibinfo{author}{F.~Husson}, \bibinfo{author}{S.~L{\^e}},
  \bibinfo{author}{J.~Pag{\`e}s}, \bibinfo{title}{{Exploratory Multivariate
  Analysis by Example Using R }}, \bibinfo{publisher}{CRC Press},
  \bibinfo{year}{2011}.
%Type = Article
\bibitem[{Iwinski(1988)}]{it2}
\bibinfo{author}{T.B. Iwinski}, \bibinfo{title}{{Rough Orders and Rough
  Concepts}}, \bibinfo{journal}{Bull. Pol. Acad. Sci (Math)}
  \bibinfo{volume}{(3--4)} (\bibinfo{year}{1988}) \bibinfo{pages}{187--192}.
%Type = Book
\bibitem[{Jacobs et~al.(2016)Jacobs, Renandya and Power}]{jrp2016}
\bibinfo{author}{G.M. Jacobs}, \bibinfo{author}{W.A. Renandya},
  \bibinfo{author}{A.~Power}, \bibinfo{title}{{Simple, Powerful Strategies for
  Student Centered Learning}}, {Springer Briefs in Education},
  \bibinfo{publisher}{Springer Nature}, \bibinfo{year}{2016}.
%Type = Article
\bibitem[{Krahmer and van Deemter(2012)}]{ekkd2012}
\bibinfo{author}{E.~Krahmer}, \bibinfo{author}{K.~van Deemter},
  \bibinfo{title}{{Computational Generation of Referring Expressions: A
  Survey}}, \bibinfo{journal}{Computational Linguistics} \bibinfo{volume}{38}
  (\bibinfo{year}{2012}) \bibinfo{pages}{173--218}.
%Type = Article
\bibitem[{Liang et~al.(2003)Liang, Shi and Li}]{jzd2003}
\bibinfo{author}{J.~Liang}, \bibinfo{author}{Z.~Shi}, \bibinfo{author}{D.~Li},
  \bibinfo{title}{{Applications of Inclusion Degree in Rough Set Theory}},
  \bibinfo{journal}{International Journal of Computational Cognition: YangSky}
  \bibinfo{volume}{1} (\bibinfo{year}{2003}) \bibinfo{pages}{67--78}.
%Type = Article
\bibitem[{Lin(2009)}]{tyl}
\bibinfo{author}{T.Y. Lin}, \bibinfo{title}{{Granular Computing-1: The Concept
  of Granulation and Its Formal Model}}, \bibinfo{journal}{Int. J. Granular
  Computing, Rough Sets and Int Systems} \bibinfo{volume}{1}
  (\bibinfo{year}{2009}) \bibinfo{pages}{21--42}.
%Type = Book
\bibitem[{Ljapin(1996)}]{lj}
\bibinfo{author}{E.S. Ljapin}, \bibinfo{title}{{Partial Algebras and Their
  Applications}}, \bibinfo{publisher}{Academic, Kluwer}, \bibinfo{year}{1996}.
%Type = Article
\bibitem[{Mani(2005)}]{am3}
\bibinfo{author}{A.~Mani}, \bibinfo{title}{{Super Rough Semantics}},
  \bibinfo{journal}{Fundamenta Informaticae} \bibinfo{volume}{65}
  (\bibinfo{year}{2005}) \bibinfo{pages}{249--261}.
%Type = Incollection
\bibitem[{Mani(2008)}]{am24}
\bibinfo{author}{A.~Mani}, \bibinfo{title}{{Esoteric Rough Set Theory-Algebraic
  Semantics of a Generalized VPRS and VPRFS}}, in:
  \bibinfo{editor}{A.~Skowron}, \bibinfo{editor}{J.F. Peters} (Eds.),
  \bibinfo{booktitle}{{Transactions on Rough Sets, LNCS 5084}}, volume
  \bibinfo{volume}{VIII}, \bibinfo{publisher}{Springer Verlag},
  \bibinfo{year}{2008}, pp. \bibinfo{pages}{182--231}.
%Type = Article
\bibitem[{Mani(2009)}]{am105}
\bibinfo{author}{A.~Mani}, \bibinfo{title}{{Algebraic Semantics of
  Similarity-Based Bitten Rough Set Theory}}, \bibinfo{journal}{Fundamenta
  Informaticae} \bibinfo{volume}{97} (\bibinfo{year}{2009})
  \bibinfo{pages}{177--197}.
%Type = Article
\bibitem[{Mani(2011)}]{am99}
\bibinfo{author}{A.~Mani}, \bibinfo{title}{{Choice Inclusive General Rough
  Semantics}}, \bibinfo{journal}{Information Sciences} \bibinfo{volume}{181}
  (\bibinfo{year}{2011}) \bibinfo{pages}{1097--1115}.
%Type = Article
\bibitem[{Mani(2012)}]{am240}
\bibinfo{author}{A.~Mani}, \bibinfo{title}{{Dialectics of Counting and The
  Mathematics of Vagueness}}, \bibinfo{journal}{Transactions on Rough Sets}
  \bibinfo{volume}{XV} (\bibinfo{year}{2012}) \bibinfo{pages}{122--180}.
%Type = Inproceedings
\bibitem[{Mani(2013)}]{am3600}
\bibinfo{author}{A.~Mani}, \bibinfo{title}{{Contamination-Free Measures and
  Algebraic Operations}}, \bibinfo{journal}{IEEEXplore}, in:
  \bibinfo{booktitle}{{Fuzzy Systems (FUZZ), 2013 IEEE International Conference
  on}}, \bibinfo{publisher}{IEEE}, \bibinfo{year}{2013}, pp.
  \bibinfo{pages}{1--8}.
%Type = Article
\bibitem[{Mani(2014)}]{am3930}
\bibinfo{author}{A.~Mani}, \bibinfo{title}{{Ontology, Rough Y-Systems and
  Dependence}}, \bibinfo{journal}{Internat. J of Comp. Sci. and Appl.}
  \bibinfo{volume}{11} (\bibinfo{year}{2014}) \bibinfo{pages}{114--136}.
  \bibinfo{note}{Special Issue of IJCSA on Computational Intelligence}.
%Type = Article
\bibitem[{Mani(2016{\natexlab{a}})}]{am9501}
\bibinfo{author}{A.~Mani}, \bibinfo{title}{{Algebraic Semantics of
  Proto-Transitive Rough Sets}}, \bibinfo{journal}{Transactions on Rough Sets}
  \bibinfo{volume}{XX} (\bibinfo{year}{2016}{\natexlab{a}})
  \bibinfo{pages}{51--108}.
%Type = Book
\bibitem[{Mani(2016{\natexlab{b}})}]{amdsc2016}
\bibinfo{author}{A.~Mani}, \bibinfo{title}{{Granular Foundations of the
  Mathematics of Vagueness, Algebraic Semantics and Knowledge Interpretation}},
  \bibinfo{publisher}{University of Calcutta},
  \bibinfo{year}{2016}{\natexlab{b}}.
%Type = Article
\bibitem[{Mani(2016{\natexlab{c}})}]{am9411}
\bibinfo{author}{A.~Mani}, \bibinfo{title}{{Probabilities, Dependence and Rough
  Membership Functions}}, \bibinfo{journal}{International Journal of Computers
  and Applications} \bibinfo{volume}{39} (\bibinfo{year}{2016}{\natexlab{c}})
  \bibinfo{pages}{17--35}.
%Type = Incollection
\bibitem[{Mani(2017{\natexlab{a}})}]{am9006}
\bibinfo{author}{A.~Mani}, \bibinfo{title}{{Approximations From Anywhere and
  General Rough Sets}}, in: \bibinfo{editor}{L.~Polkowski}, et~al. (Eds.),
  \bibinfo{booktitle}{{Rough Sets-2, IJCRS,2017}}, {LNAI 10314},
  \bibinfo{publisher}{Springer International},
  \bibinfo{year}{2017}{\natexlab{a}}, pp. \bibinfo{pages}{3--22}.
%Type = Incollection
\bibitem[{Mani(2017{\natexlab{b}})}]{am9204}
\bibinfo{author}{A.~Mani}, \bibinfo{title}{{Generalized Ideals and Co-Granular
  Rough Sets}}, in: \bibinfo{editor}{L.~Polkowski}, et~al. (Eds.),
  \bibinfo{booktitle}{{Rough Sets, Part 2, IJCRS,2017 }}, {LNAI 10314},
  \bibinfo{publisher}{Springer International},
  \bibinfo{year}{2017}{\natexlab{b}}, pp. \bibinfo{pages}{23--42}.
%Type = Incollection
\bibitem[{Mani(2017{\natexlab{c}})}]{am9114}
\bibinfo{author}{A.~Mani}, \bibinfo{title}{{Knowledge and Consequence in AC
  Semantics for General Rough Sets }}, in: \bibinfo{editor}{G.~Wang}, et~al.
  (Eds.), \bibinfo{booktitle}{{Thriving Rough Sets}}, volume
  \bibinfo{volume}{708} of \textit{\bibinfo{series}{{Studies in Computational
  Intelligence Series}}}, \bibinfo{publisher}{Springer International},
  \bibinfo{year}{2017}{\natexlab{c}}, pp. \bibinfo{pages}{237--268}.
%Type = Incollection
\bibitem[{Mani(2018{\natexlab{a}})}]{am501}
\bibinfo{author}{A.~Mani}, \bibinfo{title}{{Algebraic Methods for Granular
  Rough Sets}}, in: \bibinfo{editor}{A.~Mani},
  \bibinfo{editor}{I.~D{\"u}ntsch}, \bibinfo{editor}{G.~Cattaneo} (Eds.),
  \bibinfo{booktitle}{{Algebraic Methods in General Rough Sets}}, {Trends in
  Mathematics}, \bibinfo{publisher}{Birkhauser Basel},
  \bibinfo{year}{2018}{\natexlab{a}}, pp. \bibinfo{pages}{157--336}.
%Type = Article
\bibitem[{Mani(2018{\natexlab{b}})}]{am9969}
\bibinfo{author}{A.~Mani}, \bibinfo{title}{{Dialectical Rough Sets, Parthood
  and Figures of Opposition-I}}, \bibinfo{journal}{Transactions on Rough Sets}
  \bibinfo{volume}{XXI} (\bibinfo{year}{2018}{\natexlab{b}})
  \bibinfo{pages}{96--141}.
%Type = Incollection
\bibitem[{Mani(2018{\natexlab{c}})}]{am5019}
\bibinfo{author}{A.~Mani}, \bibinfo{title}{{Representation, Duality and
  Beyond}}, in: \bibinfo{editor}{A.~Mani}, \bibinfo{editor}{I.~D{\"u}ntsch},
  \bibinfo{editor}{G.~Cattaneo} (Eds.), \bibinfo{booktitle}{{Algebraic Methods
  in General Rough Sets}}, {Trends in Mathematics},
  \bibinfo{publisher}{Birkhauser Basel}, \bibinfo{year}{2018}{\natexlab{c}},
  pp. \bibinfo{pages}{459--552}.
%Type = Article
\bibitem[{Mao et~al.(2019)Mao, Hu and Yao}]{hmy2019}
\bibinfo{author}{H.~Mao}, \bibinfo{author}{M.~Hu}, \bibinfo{author}{Y.Y. Yao},
  \bibinfo{title}{{Algebraic Approaches To Granular Computing}},
  \bibinfo{journal}{Granular Computing}  (\bibinfo{year}{2019})
  \bibinfo{pages}{1--13}.
%Type = Book
\bibitem[{Moshkov et~al.(2008)Moshkov, Piliszczuk and Zielsko}]{mopizi}
\bibinfo{author}{M.~Moshkov}, \bibinfo{author}{M.~Piliszczuk},
  \bibinfo{author}{B.~Zielsko}, \bibinfo{title}{{Partial Covers, Reducts and
  Decision Rules in Rough Sets}}, volume \bibinfo{volume}{145} of
  \textit{\bibinfo{series}{{Studies in Computational Intelligence, Vol145}}},
  \bibinfo{publisher}{Springer Verlag}, \bibinfo{year}{2008}.
%Type = Article
\bibitem[{Mundici(1984)}]{md}
\bibinfo{author}{D.~Mundici}, \bibinfo{title}{{Generalization of Abstract Model
  Theory}}, \bibinfo{journal}{Fundamenta Math} \bibinfo{volume}{124}
  (\bibinfo{year}{1984}) \bibinfo{pages}{1--25}.
%Type = Article
\bibitem[{Osmialowski and Polkowski(2012)}]{olp2012}
\bibinfo{author}{P.~Osmialowski}, \bibinfo{author}{L.~Polkowski},
  \bibinfo{title}{{Spatial Reasoning Based On Rough Mereology}},
  \bibinfo{journal}{Transactions on Rough Sets} \bibinfo{volume}{XII}
  (\bibinfo{year}{2012}) \bibinfo{pages}{143--169}.
%Type = Book
\bibitem[{Pagliani and Chakraborty(2008)}]{ppm2}
\bibinfo{author}{P.~Pagliani}, \bibinfo{author}{M.~Chakraborty},
  \bibinfo{title}{{A Geometry of Approximation: Rough Set Theory: Logic,
  Algebra and Topology of Conceptual Patterns}}, \bibinfo{publisher}{Springer},
  \bibinfo{address}{Berlin}, \bibinfo{year}{2008}.
%Type = Book
\bibitem[{Pawlak(1991)}]{zpb}
\bibinfo{author}{Z.~Pawlak}, \bibinfo{title}{{Rough Sets: Theoretical Aspects
  of Reasoning About Data}}, \bibinfo{publisher}{Kluwer Academic Publishers},
  \bibinfo{address}{Dodrecht}, \bibinfo{year}{1991}.
%Type = Article
\bibitem[{Pawlak and Skowron(2007)}]{zpsk07}
\bibinfo{author}{Z.~Pawlak}, \bibinfo{author}{A.~Skowron},
  \bibinfo{title}{{Rough Sets and Boolean Reasoning}},
  \bibinfo{journal}{Information Sciences} \bibinfo{volume}{77}
  (\bibinfo{year}{2007}) \bibinfo{pages}{41--73}.
%Type = Incollection
\bibitem[{Pilarek and Or{\l}owska(2008)}]{joanna2008}
\bibinfo{author}{J.G. Pilarek}, \bibinfo{author}{E.~Or{\l}owska},
  \bibinfo{title}{{Logics of Similarity and Dual Tableux: A Survey}}, in:
  \bibinfo{editor}{G.D. Riccia}, et~al. (Eds.),
  \bibinfo{booktitle}{{Preferences And Similarities}},
  \bibinfo{publisher}{CISM, Springer}, \bibinfo{year}{2008}, pp.
  \bibinfo{pages}{129--160}.
%Type = Book
\bibitem[{Polkowski(2011)}]{lp2011}
\bibinfo{author}{L.~Polkowski}, \bibinfo{title}{{Approximate Reasoning by
  Parts}}, \bibinfo{publisher}{Springer Verlag}, \bibinfo{year}{2011}.
%Type = Article
\bibitem[{Polkowski and Semeniuk-Polkowska(2010)}]{lpmsp2010}
\bibinfo{author}{L.~Polkowski}, \bibinfo{author}{M.~Semeniuk-Polkowska},
  \bibinfo{title}{{Granular Rough Mereological Logics with Applications to
  Dependencies in Information and Decision Systems}},
  \bibinfo{journal}{Transactions on Rough Sets} \bibinfo{volume}{XII}
  (\bibinfo{year}{2010}) \bibinfo{pages}{1--20}.
%Type = Article
\bibitem[{Polkowski and Skowron(1996)}]{ps3}
\bibinfo{author}{L.~Polkowski}, \bibinfo{author}{A.~Skowron},
  \bibinfo{title}{{Rough Mereology: A New Paradigm for Approximate Reasoning}},
  \bibinfo{journal}{Internat. J. Appr. Reasoning} \bibinfo{volume}{15}
  (\bibinfo{year}{1996}) \bibinfo{pages}{333--365}.
%Type = Incollection
\bibitem[{Schubert(2015)}]{sepsl}
\bibinfo{author}{L.~Schubert}, \bibinfo{title}{{Computational Linguistics}},
  in: \bibinfo{editor}{E.~Zalta} (Ed.), \bibinfo{booktitle}{{Stanford
  Encyclopedia of Philosophy}}, \bibinfo{publisher}{Metaphysics Research Lab,
  Stanford University}, \bibinfo{year}{2015}.
%Type = Incollection
\bibitem[{Seibt(2017)}]{seibtj2015}
\bibinfo{author}{J.~Seibt}, \bibinfo{title}{{Transitivity}}, in:
  \bibinfo{editor}{H.~Burkhardt}, \bibinfo{editor}{J.~Seibt},
  \bibinfo{editor}{G.~Imaguire}, \bibinfo{editor}{S.~Gerogiorgakis} (Eds.),
  \bibinfo{booktitle}{{Handbook of Mereology}}, \bibinfo{publisher}{Philosophia
  Verlag}, \bibinfo{address}{Germany}, \bibinfo{year}{2017}, pp.
  \bibinfo{pages}{570--579}.
%Type = Incollection
\bibitem[{Skowron and Grzymala-Busse(1994)}]{skg1994}
\bibinfo{author}{A.~Skowron}, \bibinfo{author}{J.~Grzymala-Busse},
  \bibinfo{title}{{From Rough Set Theory to Evidence Theory}}, in:
  \bibinfo{editor}{R.~Yager}, et~al. (Eds.), \bibinfo{booktitle}{{Advances in
  the Dempster--Shafer Theory of Evidence}}, \bibinfo{publisher}{Wiley},
  \bibinfo{year}{1994}, pp. \bibinfo{pages}{193--236}.
%Type = Article
\bibitem[{Skowron and Jankowski(2016)}]{skaj2016}
\bibinfo{author}{A.~Skowron}, \bibinfo{author}{A.~Jankowski},
  \bibinfo{title}{{Rough Sets and Interactive Granular Computing}},
  \bibinfo{journal}{Fundamenta Informaticae} \bibinfo{volume}{147}
  (\bibinfo{year}{2016}) \bibinfo{pages}{371--385}.
%Type = Article
\bibitem[{Skowron et~al.(2016)Skowron, Jankowski and Dutta}]{skajsd2016}
\bibinfo{author}{A.~Skowron}, \bibinfo{author}{A.~Jankowski},
  \bibinfo{author}{S.~Dutta}, \bibinfo{title}{{Interactive granular
  computing}}, \bibinfo{journal}{Granular Computing} \bibinfo{volume}{1}
  (\bibinfo{year}{2016}) \bibinfo{pages}{95--113}.
%Type = Article
\bibitem[{Skowron and Stepaniuk(2010)}]{ss2010}
\bibinfo{author}{A.~Skowron}, \bibinfo{author}{J.~Stepaniuk},
  \bibinfo{title}{{Approximation Spaces in Rough-Granular Computing}},
  \bibinfo{journal}{Fundamenta Informaticae} \bibinfo{volume}{100}
  (\bibinfo{year}{2010}) \bibinfo{pages}{141--157}.
%Type = Article
\bibitem[{Skowron and Stepaniuk(1996)}]{ss1}
\bibinfo{author}{A.~Skowron}, \bibinfo{author}{O.~Stepaniuk},
  \bibinfo{title}{{Tolerance Approximation Spaces}},
  \bibinfo{journal}{Fundamenta Informaticae} \bibinfo{volume}{27}
  (\bibinfo{year}{1996}) \bibinfo{pages}{245--253}.
%Type = Inproceedings
\bibitem[{{\'S}l{\c e}zak(2006{\natexlab{a}})}]{sdrskt06}
\bibinfo{author}{D.~{\'S}l{\c e}zak}, \bibinfo{title}{{Association Reducts:
  Boolean Representation}}, in: \bibinfo{editor}{G.W. et. al} (Ed.),
  \bibinfo{booktitle}{{RSKT 2006}}, {LNAI 4062}, pp. \bibinfo{pages}{305--312}.
%Type = Incollection
\bibitem[{{\'S}l{\c e}zak(2006{\natexlab{b}})}]{sl2006}
\bibinfo{author}{D.~{\'S}l{\c e}zak}, \bibinfo{title}{{Rough Sets and Bayes
  Factor}}, in: \bibinfo{editor}{J.F. Peters}, \bibinfo{editor}{A.~Skowron}
  (Eds.), \bibinfo{booktitle}{{Transactions on Rough Sets III}}, {LNCS 3400},
  \bibinfo{publisher}{Springer Verlag}, \bibinfo{year}{2006}{\natexlab{b}}, pp.
  \bibinfo{pages}{202--229}.
%Type = Inproceedings
\bibitem[{Stepaniuk(1998)}]{js1998}
\bibinfo{author}{J.~Stepaniuk}, \bibinfo{title}{{Approximation Spaces, Reducts
  and Representatives}}, in: \bibinfo{editor}{L.~Polkowski},
  \bibinfo{editor}{A.~Skowron} (Eds.), \bibinfo{booktitle}{{RSKT 1998}},
  volume~\bibinfo{volume}{2}, pp. \bibinfo{pages}{109--126}.
%Type = Book
\bibitem[{Stepaniuk(2009)}]{js09}
\bibinfo{author}{J.~Stepaniuk}, \bibinfo{title}{{Rough-Granular Computing in
  Knowledge Discovery and Data Mining}}, {Studies in Computational
  Intelligence,Volume 152}, \bibinfo{publisher}{Springer-Verlag},
  \bibinfo{year}{2009}.
%Type = Book
\bibitem[{Stolzer et~al.(2008)Stolzer, Halford and Goglia}]{sms2008}
\bibinfo{author}{A.J. Stolzer}, \bibinfo{author}{C.D. Halford},
  \bibinfo{author}{J.J. Goglia}, \bibinfo{title}{{Safety Management Systems in
  Aviation}}, {Ashgate Studies in Human Factors for Flight Operations},
  \bibinfo{publisher}{Ashgate}, \bibinfo{edition}{1} edition,
  \bibinfo{year}{2008}.
%Type = Article
\bibitem[{Syau et~al.(2017)Syau, Lin and Liau}]{yec2017}
\bibinfo{author}{Y.R. Syau}, \bibinfo{author}{E.B. Lin}, \bibinfo{author}{C.j.
  Liau}, \bibinfo{title}{{Neighborhood Systems and Variable Precision
  Generalized Rough Sets}}, \bibinfo{journal}{Fundamenta Informaticae}
  \bibinfo{volume}{153} (\bibinfo{year}{2017}) \bibinfo{pages}{271--290}.
%Type = Phdthesis
\bibitem[{Urbaniak(2008)}]{ur}
\bibinfo{author}{R.~Urbaniak}, \bibinfo{title}{{Lesniewski's Systems of Logic
  and Mereology; History and Re-Evaluation}}, Ph.D. thesis, Department of
  Philosophy, Univ of Calgary, \bibinfo{year}{2008}.
%Type = Article
\bibitem[{Varzi(1996)}]{av}
\bibinfo{author}{A.~Varzi}, \bibinfo{title}{{Parts, Wholes and Part-Whole
  Relations: The Prospects of Mereotopology}}, \bibinfo{journal}{Data and
  Knowledge Engineering} \bibinfo{volume}{20} (\bibinfo{year}{1996})
  \bibinfo{pages}{259--286}.
%Type = Article
\bibitem[{Vieu(2007)}]{vie}
\bibinfo{author}{L.~Vieu}, \bibinfo{title}{{On The Transitivity of Functional
  Parthood}}, \bibinfo{journal}{Applied Ontology} \bibinfo{volume}{1}
  (\bibinfo{year}{2007}) \bibinfo{pages}{147--155}.
%Type = Article
\bibitem[{Wu et~al.(2002)Wu, Leung and Zhang}]{wu2002}
\bibinfo{author}{W.Z. Wu}, \bibinfo{author}{Y.~Leung},
  \bibinfo{author}{W.~Zhang}, \bibinfo{title}{{Connections Between Rough Set
  Theory and Dempster--shafer Theory of Evidence}}, \bibinfo{journal}{Int. J.
  General Systems} \bibinfo{volume}{31} (\bibinfo{year}{2002})
  \bibinfo{pages}{405--430}.
%Type = Article
\bibitem[{Yao et~al.(2019)Yao, Miao, Pedrycz, Zhang and Zhang}]{ndwhz2019}
\bibinfo{author}{N.~Yao}, \bibinfo{author}{D.~Miao},
  \bibinfo{author}{W.~Pedrycz}, \bibinfo{author}{H.~Zhang},
  \bibinfo{author}{Z.~Zhang}, \bibinfo{title}{{Causality Measures and Analysis:
  A Rough Set Framework}}, \bibinfo{journal}{Expert Systems With Applications}
  (\bibinfo{year}{2019}) \bibinfo{pages}{1--33}.
%Type = Article
\bibitem[{Yao(2010)}]{yy10}
\bibinfo{author}{Y.Y. Yao}, \bibinfo{title}{{Three-Way Decisions With
  Probabilistic Rough Sets}}, \bibinfo{journal}{Information Sciences}
  \bibinfo{volume}{180} (\bibinfo{year}{2010}) \bibinfo{pages}{341--353}.
%Type = Article
\bibitem[{Yao(2015)}]{yy2015}
\bibinfo{author}{Y.Y. Yao}, \bibinfo{title}{{The Two Sides of The Theory of
  Rough Sets}}, \bibinfo{journal}{Knowledge-Based Systems} \bibinfo{volume}{80}
  (\bibinfo{year}{2015}) \bibinfo{pages}{67--77}.
%Type = Article
\bibitem[{Yao et~al.(2012)Yao, Zhang and Miao}]{yzm2012}
\bibinfo{author}{Y.Y. Yao}, \bibinfo{author}{N.~Zhang},
  \bibinfo{author}{D.~Miao}, \bibinfo{title}{{Set-Theoretic Approaches To
  Granular Computing}}, \bibinfo{journal}{Fundamenta Informaticae}
  \bibinfo{volume}{115} (\bibinfo{year}{2012}) \bibinfo{pages}{247--264}.
%Type = Article
\bibitem[{Ziarko(1993)}]{zw}
\bibinfo{author}{W.~Ziarko}, \bibinfo{title}{{Variable Precision Rough Set
  Model}}, \bibinfo{journal}{J. of Computer and System Sciences}
  \bibinfo{volume}{46} (\bibinfo{year}{1993}) \bibinfo{pages}{39--59}.

\end{thebibliography}
\end{document}